\documentclass[journal,twocolumn,twoside]{IEEEtran}

\IEEEoverridecommandlockouts
\usepackage{amsfonts}
\usepackage{color}
\usepackage{graphicx}
\usepackage[dvips]{epsfig}
\usepackage{graphics} 
\usepackage{times} 
\usepackage[cmex10]{amsmath} 
\usepackage{amssymb}  
\usepackage{cite}
\usepackage{multirow}
\usepackage[tight,footnotesize]{subfigure}
\usepackage{algorithm}

\def\norm #1{\left\|#1\right\|}

\def\twon #1{\left\|#1\right\|_2}
\def\onen #1{\left\|#1\right\|_1}

\def\frobn #1{\left\|#1\right\|_{\text{F}}}

\def\atomn #1{\left\|#1\right\|_{\cA}}

\def\abs #1{\left|#1\right|}

\def\st{\text{subject to }}

\def\bC{\mathbb{C}}

\def\bR{\mathbb{R}}

\def\bT{\mathbb{T}}

\def\m #1{\boldsymbol{#1}}

\def\cA{\mathcal{A}}

\def\cK{\mathcal{K}}
\def\cL{\mathcal{L}}
\def\cM{\mathcal{M}}

\def\cQ{\mathcal{Q}}

\def\cS{\mathcal{S}}

\def\bee{\begin{equation}}
\def\ene{\end{equation}}

\def\beq{\begin{eqnarray}}
\def\enq{\end{eqnarray}}

\newtheorem{rem}{Remark}

\newtheorem{thm}{Theorem}
\newtheorem{prop}{Proposition}

\def\equ #1{\begin{equation}#1\end{equation}}

\def\sbra #1{\left(#1\right)}
\def\mbra #1{\left[#1\right]}
\def\lbra #1{\left\{#1\right\}}
\def\diag #1{\text{diag}#1}
\def\tr #1{\text{tr}#1}

\def\rank #1{\text{rank}#1}
\def\st {\text{ subject to }}

\title{Enhancing Sparsity and Resolution via Reweighted Atomic Norm Minimization}
\author{Zai Yang, {\em Member, IEEE}, and Lihua Xie, {\em Fellow, IEEE}
\thanks{Manuscript August 2014; revised March and August 2015; accepted October 2015. Parts of this paper were presented at the 2015 IEEE International Conference on Acoustics, Speech and Signal Processing (ICASSP), Brisbane, Australia, April 2015 \cite{yang2015achieving}.

The authors are with the School of Electrical and Electronic Engineering, Nanyang Technological University, 639798,
Singapore (e-mail: \{yangzai, elhxie\}@ntu.edu.sg).}}

\begin{document}
\maketitle

\begin{abstract}
The mathematical theory of super-resolution developed recently by Cand\`{e}s and Fernandes-Granda states that a continuous, sparse frequency spectrum can be recovered with infinite precision via a (convex) atomic norm technique given a set of uniform time-space samples. This theory was then extended to the cases of partial/compressive samples and/or multiple measurement vectors via atomic norm minimization (ANM), known as off-grid/continuous compressed sensing (CCS). However, a major problem of existing atomic norm methods is that the frequencies can be recovered only if they are sufficiently separated, prohibiting commonly known high resolution. In this paper, a novel (nonconvex) sparse metric is proposed that promotes sparsity to a greater extent than the atomic norm. Using this metric an optimization problem is formulated and a locally convergent iterative algorithm is implemented. The algorithm iteratively carries out ANM with a sound reweighting strategy which enhances sparsity and resolution, and is termed as reweighted atomic-norm minimization (RAM). Extensive numerical simulations are carried out to demonstrate the advantageous performance of RAM with application to direction of arrival (DOA) estimation.
\end{abstract}


\section{Introduction}
Compressed sensing (CS) \cite{candes2006compressive,donoho2006compressed} refers to a technique of reconstructing a high dimensional signal from far fewer samples and has brought significant impact on signal processing and information theory in the past decade. In conventional wisdom, the signal of interest needs to be sparse under a finite discrete dictionary for successful reconstruction, which limits its applications, for example, to array processing, radar and sonar, where the dictionary is typically specified by one or more continuous parameters. In this paper, we are concerned about a compressed sensing problem with a continuous dictionary which arises in line spectral estimation and array processing \cite{krim1996two,stoica2005spectral}. In particular, we are interested in recovering $L$ discrete sinusoidal signals which compose the data matrix $\m{Y}^o\in\bC^{N\times L}$ with its $(j,t)$th element (corrupted by noise in practice)
\equ{y_{jt}^o=\sum_{k=1}^K s_{kt}e^{i2\pi (j-1) f_k}, \quad \sbra{j,t}\in \mbra{N}\times\mbra{L}, \label{formu:observmodel2}}
where $i=\sqrt{-1}$, $f_k\in\bT\triangleq\left[0,1\right]$, $s_{kt}\in\bC$ and $\mbra{N}=\lbra{1,2,\dots,N}$.
This means that each column of $\m{Y}^o$ is superimposed by $K$ discrete sinusoids with frequencies $\lbra{f_k}$ and amplitudes $\lbra{s_{kt}}$. To recover $\m{Y}^o$ (and the frequencies in many applications), however, we are only given partial/compressive samples on its rows indexed by $\m{\Omega} \subset \mbra{N}$ (of size $M < N$), denoted by $\m{Y}_{\m{\Omega}}^o$. This problem is referred to as off-grid or continuous compressed sensing (CCS) according to \cite{tang2012compressed,yang2014continuous} differing from the existing CS framework in the sense that every frequency $f_k$ can take any continuous value in $\bT$ rather than constrained on a finite discrete grid.

The CCS problem in the case of $L=1$ (a.k.a. the single-measurement-vector (SMV) case) is usually known as line spectral estimation in which frequency recovery though is of main interest. The use of compressive data can lead to efficient sampling and/or energy saving. It also can be caused by data missing due to adversary environmental effects. The multiple-measurement-vector (MMV) case with $L>1$ is common in array processing where one estimates directions of a few narrowband sources using outputs of an antenna array. Readers are referred to \cite{krim1996two,stoica2005spectral} for derivation of the model in \eqref{formu:observmodel2}. Therein $\m{Y}^o$ consists of outputs of a virtual $N$-element uniform linear array (ULA), in which adjacent antennas are spaced by half a wavelength, over $L$ time snapshots. In particular, each column of $\m{Y}^o$ corresponds to one snapshot of the ULA and each row consists of outputs of a single antenna. The fact that we have only access to $\m{Y}_{\m{\Omega}}^o$ means that we actually use a sparse linear array (SLA) that is obtained by retaining the antennas of the ULA indexed by $\m{\Omega}$. Therefore, the index set $\m{\Omega}$ refers to geometry of the SLA and a smaller $M$ means use of fewer antennas (note that SLAs are common in practice for obtaining a large aperture from a limited number of antennas, see, e.g., \cite{linebarger1993difference} and the references therein). Each frequency component corresponds to one source. The value of $f_k$ uniquely determines the direction of source $k$, and vice versa. Consequently, the problem of direction of arrival (DOA) estimation using a SLA $\m{\Omega}$ is exactly the frequency estimation problem in CCS given the measurement matrix $\m{Y}_{\m{\Omega}}^o$.

Due to its connections to line spectral estimation and DOA estimation, studies of the CCS problem have a long history while frequency estimation has been mainly focused on. Well known conventional methods include periodogram (or beamforming), Capon's beamforming and subspace methods like MUSIC (see the review in \cite{stoica2005spectral}). Periodogram suffers from the so-called leakage problem and the Fourier resolution limit of $\frac{1}{N}$ even in the full data case when $M=N$ \cite{stoica2005spectral}. It therefore has difficulties in resolving two closely spaced frequencies. The situation becomes even worse in the compressive data case. Capon's beamforming and MUSIC are \emph{high resolution} methods in the sense that they can break the aforementioned resolution limit. Since they are covariance-based methods sufficient snapshots are required to estimate the data covariance. Moreover, they are sensitive to source correlations. With the development of sparse signal representation and later the CS concept, sparse methods have been popular in the last decade which exploit the prior knowledge that the number of frequency components $K$ is small \cite{malioutov2005sparse,hyder2010direction,stoica2011new,stoica2011spice, wei2012doa,fannjiang2012coherence,hu2013doa, duarte2013spectral,liu2013sparsity,hu2012compressed,yang2012robustly,yang2013off,hu2013fast,austin2013dynamic,tan2014joint}. In these methods, however, the frequency domain $\bT$ has to be gridded/discretized into a finite set, resulting in the grid mismatch problem that limits the estimation accuracy as well as brings challenges to the theoretical performance analysis \cite{chi2011sensitivity,yang2013off}. Though modified, off-grid estimation methods \cite{hu2012compressed,yang2012robustly,yang2013off,hu2013fast,austin2013dynamic,tan2014joint} have been implemented to alleviate these drawbacks, overall they are still based on gridding of the frequency domain.

A mathematical theory of super-resolution was recently introduced by Cand\`{e}s and Fernandes-Granda \cite{candes2013towards}. They studied frequency estimation in the SMV, full data case and proposed a \emph{gridless} convex optimization method based on the atomic norm (or the total variation norm) \cite{chandrasekaran2012convex}. In addition, they proved that the frequencies can be recovered with infinite precision in the absence of noise once they are mutually separated by at least $\frac{4}{N}$. This theory was then extended to the cases of compressive data and MMVs by Tang \emph{et al.} \cite{tang2012compressed} and the authors \cite{yang2014continuous,yang2014exact}, showing that the signal and the frequencies can be exactly recovered with high probability via atomic norm minimization (ANM) provided $M\geq O\sbra{K\ln K\ln N}$ and the same frequency separation condition holds. Other related papers include \cite{bhaskar2013atomic,candes2013super,tang2015near,yang2014discretization,yang2015gridless, azais2014spike, condat2015cadzow, chen2014robust,tan2014direction,mishra2014off,lu2015distributed}. While the atomic norm techniques completely eliminate grid mismatches of earlier grid-based sparse methods, a major problem is that the frequencies have to be sufficiently separated for successful recovery, prohibiting \emph{high resolution}.\footnote{The frequency separation $\frac{4}{N}$ is sufficient but not necessary. Empirical studies in \cite{tang2012compressed} suggest that this value is about $\frac{1}{N}$ in the SMV case, while \cite{yang2014exact} shows that it also depends on other factors like $K$, $M$ and $L$.}

In this paper, we propose a high resolution gridless sparse method for signal and frequency recovery in CCS. Our method is motivated by the formulations and properties of the atomic $\ell_0$ norm and the atomic norm in \cite{yang2014continuous,yang2014exact}. In particular, the atomic $\ell_0$ norm directly exploits sparsity and has no resolution limit but is NP hard to compute. To the contrary, as a convex relaxation the atomic norm can be efficiently computed but suffers from a resolution limit as mentioned above. We propose a novel sparse metric and theoretically show that the new metric fills the gap between the atomic $\ell_0$ norm and the atomic norm. It approaches the former under appropriate parameter setting and breaks the resolution limit. Using this sparse metric we formulate a nonconvex optimization problem for signal and frequency recovery. A locally convergent iterative algorithm is presented to solve the problem. Some further analysis shows that the algorithm iteratively carries out ANM with a sound reweighting strategy that determines preference of frequency selection based on the latest estimate and enhances sparsity and resolution. The resulting algorithm is termed as reweighted atomic-norm minimization (RAM). Extensive numerical simulations are carried out to demonstrate the performance of RAM with application to DOA estimation compared to existing art.

We note that the idea of reweighted optimization for enhancing sparsity is not new. For example, reweighted $\ell_1$ algorithms  have been introduced for discrete CS \cite{lobo2007portfolio,candes2008enhancing,wipf2010iterative,stoica2014weighted}, and reweighted trace minimization for low rank matrix recovery (LRMR) \cite{fazel2003log,mohan2012iterative}. However, it is unclear how to implement a reweighting strategy in the continuous dictionary setting until this paper. Furthermore, besides sparsity we show that the proposed reweighted algorithm enhances resolution that is of great importance in CCS.

Notations used in this paper are as follows. $\bR$ and $\bC$ denote the sets of real and complex numbers respectively. $\bT$ denotes the unit circle $\left[0,1\right]$ by identifying the beginning and the ending points. Boldface letters are reserved for vectors and matrices. For an integer $N$, $[N]\triangleq\lbra{1,\cdots,N}$. $\abs{\cdot}$ denotes cardinality of a set, amplitude of a scalar, or determinant of a squared matrix. $\onen{\cdot}$, $\twon{\cdot}$ and $\frobn{\cdot}$ denote the $\ell_1$, $\ell_2$ and Frobenius norms respectively. $\m{A}^T$ and $\m{A}^H$ are the matrix transpose and conjugate transpose of $\m{A}$ respectively. $x_j$ is the $j$th entry of a vector $\m{x}$. Unless otherwise stated, $\m{x}_{\m{\Omega}}$ and $\m{A}_{\m{\Omega}}$ respectively reserve the entries of $\m{x}$ and the rows of $\m{A}$ indexed by a set $\m{\Omega}$. For a vector $\m{x}$, $\diag\sbra{\m{x}}$ is a diagonal matrix with $\m{x}$ being its diagonal. $\rank\sbra{\cdot}$ denotes the rank and $\tr\sbra{\cdot}$ the trace. $\m{A}\geq\m{0}$ means that $\m{A}$ is positive semidefinite (PSD).

The rest of the paper is organized as follows. Section \ref{sec:preliminary} revisits preliminary gridless sparse methods that motivate this paper. Section \ref{sec:novelmetric} presents a novel sparse metric for signal and frequency recovery. Section \ref{sec:RAM} introduces the RAM algorithm. Section \ref{sec:implementation} presents some algorithm implementation strategies for accuracy and speed considerations. Section \ref{sec:simulation} provides extensive numerical simulations to demonstrate the performance of RAM. Section \ref{sec:conclusion} concludes this paper.

\section{Preliminary Gridless Sparse Methods by Exploiting Sparsity} \label{sec:preliminary}

Unless otherwise stated, we assume in this paper that the observed data $\m{Y}_{\m{\Omega}}^o$ is contaminated by noise whose Frobenius norm is bounded by $\eta\geq0$. It is clear that $\eta=0$ refers to the noiseless case. The CCS problem is solved by exploiting sparsity in the sense that the number of frequency components $K$ is small. In particular, we seek a sparse candidate $\m{Y}$ that is composed of a few frequency components and is meanwhile consistent with the observed data by imposing that $\m{Y}\in\cS$, where
\equ{\cS\triangleq\lbra{\m{Y}\in\bC^{N\times L}:\; \frobn{\m{Y}_{\m{\Omega}} - \m{Y}_{\m{\Omega}}^o}\leq\eta }. \notag}
Therefore, we first define a sparse metric of $\m{Y}$ and then optimize the metric over $\cS$ for its solution. The frequencies $\lbra{f_k}$ are estimated using the frequency components composing $\m{Y}$.

A direct sparse metric is the smallest number of frequency components composing $\m{Y}$, known as the atomic $\ell_0$ norm and denoted by $\norm{\m{Y}}_{\cA,0}$ \cite{tang2012compressed,yang2014continuous,yang2014exact}:
\equ{\norm{\m{Y}}_{\cA,0}
=\inf_{f_k,\m{s}_k}\lbra{\cK: \m{Y}=\sum_{k=1}^{\cK} \m{a}\sbra{f_k}\m{s}_k}, \label{formu:AL0}}
where $\m{a}\sbra{f}=\mbra{1,e^{i2\pi f},\dots,e^{i2\pi\sbra{N-1}f}}^T\in\bC^N$ denotes a discrete sinusoid with frequency $f\in\bT$ and $\m{s}_k\in\bC^{L\times 1}$ is the coefficient vector of the $k$th sinusoid. Following from \cite{tang2012compressed,yang2014continuous,yang2014exact}, $\norm{\m{Y}}_{\cA,0}$ can be characterized as the following rank minimization problem:
\equ{\begin{split}\norm{\m{Y}}_{\cA,0}
=&\min_{\m{u}} \rank\sbra{T\sbra{\m{u}}},\\
&\st \tr\sbra{\m{Y}^HT\sbra{\m{u}}^{-1}\m{Y}}<+\infty,\\
&\phantom{\st } T\sbra{\m{u}}\geq\m{0}.\end{split} \label{formu:atom0norm}}
Throughout this paper we use the following identity whenever $\m{R}\in \bC^{N\times N}$ is positive semidefinite:
\equ{\begin{split}
&\tr\sbra{\m{Y}^H\m{R}^{-1}\m{Y}} \\
&= \min_{\m{X}} \tr\sbra{\m{X}}, \st \begin{bmatrix}\m{X} & \m{Y}^H \\ \m{Y} & \m{R} \end{bmatrix} \geq \m{0}. \end{split} }
The first constraint in (\ref{formu:atom0norm}) imposes that $\m{Y}$ lies in the range space of a (Hermitian) Toeplitz matrix
\equ{T\sbra{\m{u}}=\begin{bmatrix}u_1 & u_2 & \cdots & u_N\\ {u}_2^H & u_1 & \cdots & u_{N-1}\\ \vdots & \vdots & \ddots & \vdots \\ {u}_N^H & {u}_{N-1}^H & \cdots & u_1\end{bmatrix}\in\bC^{N\times N}, }
where $u_j$ is the $j$th entry of $\m{u}\in\bC^N$. The frequencies composing $\m{Y}$ are encoded in $T\sbra{\m{u}}$. Once an optimizer of $\m{u}$, say $\m{u}^*$, is obtained the frequencies can be retrieved from $T\sbra{\m{u}^*}$ using the Vandermonde decomposition lemma (see, e.g., \cite{stoica2005spectral}), which states that any PSD Toeplitz matrix $T\sbra{\m{u}^*}$ can be decomposed as
\equ{T\sbra{\m{u}^*}=\sum_{k=1}^{K^*} p_k^*\m{a}\sbra{f_k^*}\m{a}\sbra{f_k^*}^H,}
where the order $K^*=\rank\sbra{T\sbra{\m{u}^*}}$ and $p_k^*>0$ (note that this decomposition is unique if $K^*<N$ and a computational method can be found in \cite[Appendix A]{yang2015gridless}). Therefore, by (\ref{formu:atom0norm}) the CCS problem is reformulated as a LRMR problem in which the matrix $T\sbra{\m{u}}$ is Toeplitz and PSD and its range space contains $\m{Y}$.

The atomic $\ell_0$ norm exploits sparsity to the greatest extent possible; however, it is nonconvex and NP-hard to compute according to the rank minimization formulation and it thus encourages computationally feasible alternatives. In this spirit, the atomic ($\ell_1$) norm, denoted by $\atomn{\m{Y}}$, is introduced as a convex relaxation of $\norm{\m{Y}}_{\cA,0}$ \cite{tang2012compressed,yang2014continuous,yang2014exact}:
\equ{\begin{split}\atomn{\m{Y}}= \inf_{f_k,\m{s}_k}\lbra{\sum_k \twon{\m{s}_k}: \m{Y}=\sum_k \m{a}\sbra{f_k}\m{s}_k}\end{split} \label{formu:atomicnorm}}
which is a continuous counterpart of the $\ell_{2,1}$ norm utilized for joint sparse recovery in discrete CS (see, e.g., \cite{malioutov2005sparse,van2010theoretical}). $\atomn{\m{Y}}$ is a norm and has the following semidefinite formulation \cite{tang2012compressed,yang2014continuous,yang2014exact}:
\equ{\begin{split}\norm{\m{Y}}_{\cA}=
&\min_{\m{u}} \frac{1}{2\sqrt{N}} \mbra{\tr\sbra{T\sbra{\m{u}}} + \tr\sbra{\m{Y}^HT\sbra{\m{u}}^{-1} \m{Y}}},\\
&\st T\sbra{\m{u}}\geq\m{0}. \end{split} \label{formu:AN_SDP}}
From the perspective of LRMR, (\ref{formu:AN_SDP}) attempts to recover the low rank matrix $T\sbra{\m{u}}$ by relaxing the pseudo rank norm in (\ref{formu:atom0norm}) to the nuclear norm (or the trace norm for a PSD matrix). Again, the frequencies are encoded in $T\sbra{\m{u}}$ and can be obtained using the Vandermonde decomposition once the optimization problem is solved within a polynomial time. The atomic norm is computationally advantageous compared to the atomic $\ell_0$ norm while it suffers from a resolution limit due to the relaxation which is not shared by the latter \cite{candes2013towards,tang2012compressed,yang2014exact}.

\section{Enhancing Sparsity and Resolution via A Novel Sparse Metric} \label{sec:novelmetric}


Inspired by the link between CCS and LRMR demonstrated above, we propose the following sparse metric of $\m{Y}$:
\equ{\begin{split}\cM^{\epsilon}\sbra{\m{Y}}=
&\min_{\m{u}} \ln\abs{T\sbra{\m{u}}+\epsilon\m{I}} + \tr\sbra{\m{Y}^HT\sbra{\m{u}}^{-1}\m{Y}},\\
&\st T\sbra{\m{u}}\geq\m{0},\end{split} \label{formu:nonconvexrelax}}
where $\epsilon>0$ is a regularization parameter that avoids the first term being $-\infty$ when $T\sbra{\m{u}}$ is rank deficient.
Note that the log-det heuristic $\log\abs{\cdot}$ has been widely used as a smooth surrogate for the rank of a PSD matrix (see, e.g., \cite{david1994algorithms,fazel2003log,mohan2012iterative}). Also, a similar logarithmic penalty has been adopted to approximate the pseudo $\ell_0$ norm for discrete sparse recovery \cite{rao1999affine,chartrand2008iteratively,candes2008enhancing}. From the perspective of LRMR, the atomic $\ell_0$ norm minimizes the number of nonzero eigenvalues of $T\sbra{\m{u}}$ while the atomic norm minimizes the sum of the eigenvalues. In contrast, the new metric $\cM^{\epsilon}\sbra{\m{Y}}$ puts penalty on $\sum_{k=1}^N\ln\abs{\lambda_k+\epsilon}$, where $\lbra{\lambda_k}_{k=1}^N$ denotes the eigenvalues. We plot the function $h(\lambda)=\ln\abs{\lambda+\epsilon}$ with different $\epsilon$'s in Fig. \ref{Fig:sparsity} together with the constant function (except at $\lambda=0$) and the identity function corresponding to the $\ell_0$ and $\ell_1$ norms respectively, where $h(\lambda)$ is translated and scaled for better illustration without altering its sparsity-enhancing property. Intuitively, $h(\lambda)$ gets close to the $\ell_1$ norm for large $\epsilon$ while it approaches the $\ell_0$ norm as $\epsilon\rightarrow0$. Therefore, we expect that the new metric $\cM^{\epsilon}\sbra{\m{Y}}$ bridges $\atomn{\m{Y}}$ and $\norm{\m{Y}}_{\cA,0}$ when $\epsilon$ varies from $+\infty$ to $0$. Formally, we have the following results.

\begin{figure}
\centering
  \includegraphics[width=3in]{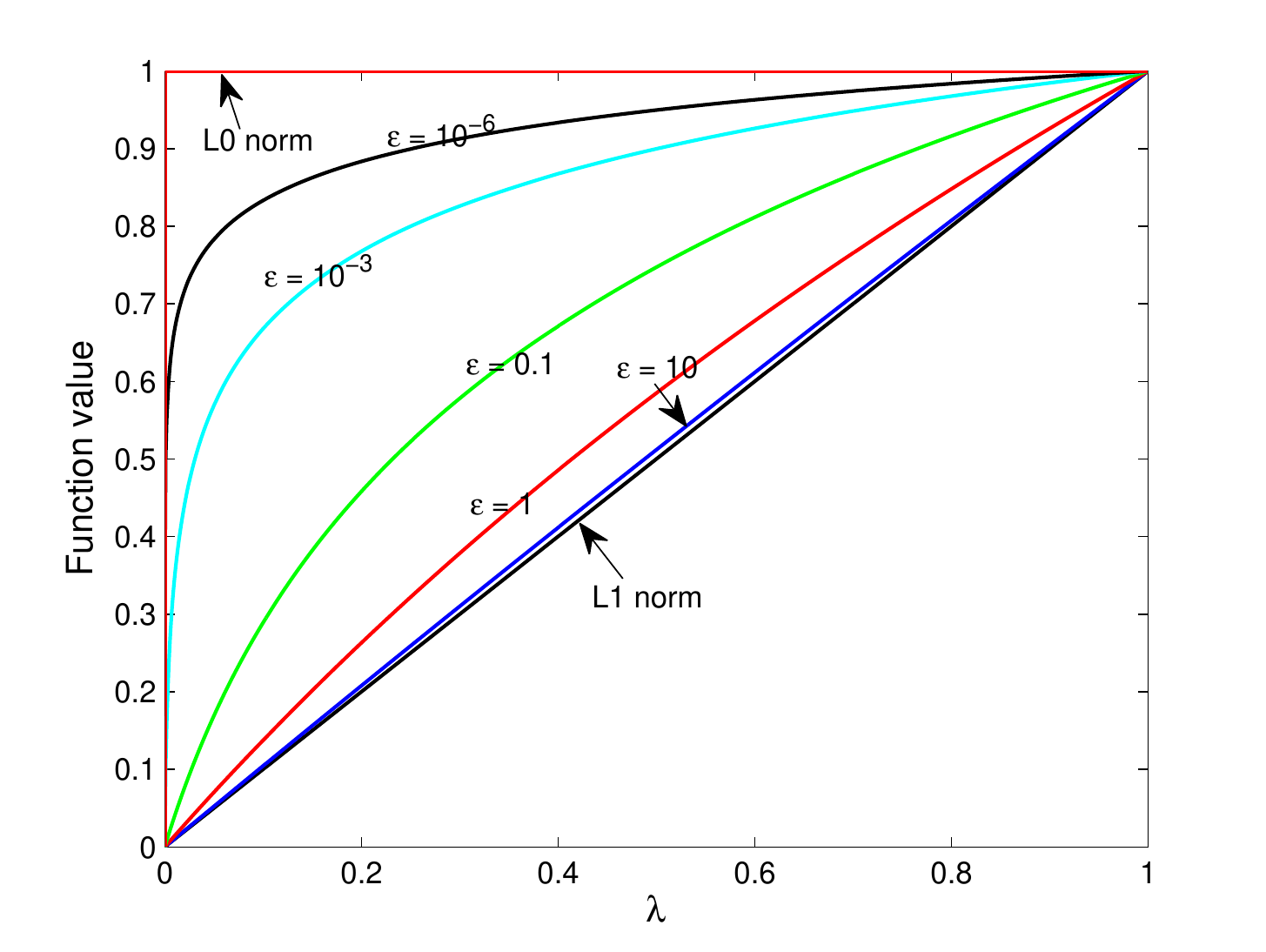}
\centering
\caption{Illustration of the sparsity-promoting property of $\cM^{\epsilon}\sbra{\cdot}$ with respect to $\epsilon$. The plotted curves include the $\ell_0$ and $\ell_1$ norms corresponding to $\norm{\cdot}_{\cA,0}$ and $\norm{\cdot}_{\cA}$ respectively, and $\ln\abs{\lambda+\epsilon}$ corresponding to $\cM^{\epsilon}\sbra{\cdot}$ with $\epsilon=10, 1, 0.1, 10^{-3}$ and $10^{-6}$. $\ln\abs{\lambda+\epsilon}$ is translated and scaled such that it equals 0 and 1 at $\lambda=0$ and 1 respectively for better illustration.} \label{Fig:sparsity}
\end{figure}

\begin{thm} Let $\epsilon\rightarrow+\infty$. Then,
\equ{\cM^{\epsilon}\sbra{\m{Y}}-N\ln\epsilon \sim 2\sqrt{N}\atomn{\m{Y}}\epsilon^{-\frac{1}{2}}, \label{formu:epsilontoinf}}
i.e., they are equivalent infinitesimals. \label{thm:epsilontoinf}
\end{thm}
\begin{proof} See Appendix \ref{sec:append:epsilontoinf}.
\end{proof}

\begin{thm} Let $r=\norm{\m{Y}}_{\cA,0}$ and $\epsilon\rightarrow0$. Then, we have the following results:
\begin{enumerate}
 \item If $r\leq N-1$, then
     \equ{\cM^{\epsilon}\sbra{\m{Y}}\sim \sbra{r-N}\ln\frac{1}{\epsilon}, \label{formu:equivinf}}
     i.e., they are equivalent infinities. Otherwise, $\cM^{\epsilon}\sbra{\m{Y}}$ is a positive constant depending only on $\m{Y}$;
 \item Let $\m{u}_{\epsilon}^*$ be the (global) optimizer of $\m{u}$ to the optimization problem in (\ref{formu:nonconvexrelax}). Then, the smallest $N-r$ eigenvalues of $T\sbra{\m{u}_{\epsilon}^*}$ are either zero or approach zero as fast as $\epsilon$;
 \item For any cluster point of $\m{u}_{\epsilon}^*$ at $\epsilon=0$, denoted by $\m{u}_0^*$, there exists an atomic decomposition $\m{Y}=\sum_{k=1}^r\m{a}\sbra{f_k}\m{s}_k$ such that $T\sbra{\m{u}_0^*}=\sum_{k=1}^r\twon{\m{s}_k}^2\m{a}\sbra{f_k}\m{a}\sbra{f_k}^H$.\footnote{$\m{u}^*$ is called a cluster point of a vector-valued function $\m{u}(x)$ at $x=x_0$ if there exists a sequence $\lbra{x_n}_{n=1}^{+\infty}$, $\lim_{n\rightarrow+\infty}x_n=x_0$, satisfying that $\lim_{n\rightarrow+\infty}\m{u}(x_n)\rightarrow \m{u}^*$.}
\end{enumerate} \label{thm:epsilontozero}
\end{thm}
\begin{proof} See Appendix \ref{sec:append:epsilontozero}.
\end{proof}

\begin{rem} Note that the term $\ln\frac{1}{\epsilon}$ in \eqref{formu:equivinf} that becomes unbounded as $\epsilon\rightarrow0$ is not problematic in the optimization problem \eqref{formu:criterion3} to introduce, since the objective function $\cM^{\epsilon}\sbra{\m{Y}}$ can be re-scaled by $\sbra{\ln\frac{1}{\epsilon}}^{-1}$ for any $\epsilon>0$ without altering the optimizer. By similar arguments we see that $\epsilon^{-\frac{1}{2}}$ in \eqref{formu:epsilontoinf} is not problematic as well.
\end{rem}

Theorem \ref{thm:epsilontoinf} shows that the new metric $\cM^{\epsilon}\sbra{\m{Y}}$ plays the same role as $\atomn{\m{Y}}$ as $\epsilon\rightarrow+\infty$, while Theorem \ref{thm:epsilontozero} states that it approaches $\norm{\m{Y}}_{\cA,0}$ as $\epsilon\rightarrow0$. Consequently, it bridges $\atomn{\m{Y}}$ and $\norm{\m{Y}}_{\cA,0}$ and is expected to enhance sparsity and resolution compared to $\atomn{\m{Y}}$. Moreover, Theorem \ref{thm:epsilontozero} characterizes the properties of the optimizer $\m{u}_{\epsilon}^*$ as $\epsilon\rightarrow0$ including the convergent speed of the smallest $N-\norm{\m{Y}}_{\cA,0}$ eigenvalues and the limiting form of $T\sbra{\m{u}^*}$ via the Vandermonde decomposition. In fact, we always observe via simulations that the smallest $N-\norm{\m{Y}}_{\cA,0}$ eigenvalues of $T\sbra{\m{u}^*}$ become zero once $\epsilon$ is appropriately small.


\begin{rem} In DOA estimation, a difficult scenario is when the source signals $\lbra{\m{s}_k}_{k=1}^K$ are highly or even completely correlated (the latter case is usually called coherent). For example, covariance-based methods like Capon's beamforming and MUSIC cannot produce satisfactory results since a faithful covariance estimate is unavailable. In contrast, the proposed sparse metric is robust to source correlations by Theorem \ref{thm:epsilontozero} in which we have not made any assumption for the sources.
\end{rem}

\begin{rem} According to Theorem \ref{thm:epsilontozero}, the solution $T\sbra{\m{u}^*}$ can be interpreted as the data covariance of $\m{Y}$ after removing correlations among the sources. \label{rem:datacov}
\end{rem}

\section{Reweighted Atomic-Norm Minimization (RAM)} \label{sec:RAM}

\subsection{A Locally Convergent Iterative Algorithm}
Using the proposed sparse metric $\cM^{\epsilon}\sbra{\m{Y}}$ we solve the following optimization problem for signal and frequency recovery:
\equ{\begin{split}
\min_{\m{Y}\in\cS}\cM^{\epsilon}\sbra{\m{Y}}, \end{split} \label{formu:criterion3}}
or equivalently,
\equ{\begin{split}
&\min_{\m{Y}\in\cS,\m{u}} \ln\abs{T\sbra{\m{u}}+\epsilon\m{I}} + \tr\sbra{\m{Y}^HT\sbra{\m{u}}^{-1}\m{Y}},\\
&\st T\sbra{\m{u}}\geq\m{0}. \end{split} \label{formu:problem}}
Note that $\ln\abs{T\sbra{\m{u}}+\epsilon\m{I}}$ is a concave function of $\m{u}$ since $\ln\abs{\m{R}}$ is a concave function of $\m{R}$ on the positive semidefinite cone \cite{boyd2004convex}. It follows that the problem in \eqref{formu:problem} is nonconvex and no efficient algorithms can guarantee to obtain the global optimizer. A popular locally convergent approach to minimization of such a concave $+$ convex function is the majorization-maximization (MM) algorithm (see, e.g., \cite{fazel2003log}). Let $\m{u}_j$ denote the $j$th iterate of the optimization variable $\m{u}$. Then, at the $\sbra{j+1}$th iteration we replace $\ln\abs{T\sbra{\m{u}}+\epsilon\m{I}}$ by its tangent plane at the current value $\m{u}=\m{u}_j$:
\equ{\begin{split}
&\ln\abs{T\sbra{\m{u}_j}+\epsilon\m{I}} + \tr\mbra{\sbra{T\sbra{\m{u}_j}+\epsilon\m{I}}^{-1}T\sbra{\m{u}-\m{u}_j}}\\
&= \tr\mbra{\sbra{T\sbra{\m{u}_j}+\epsilon\m{I}}^{-1}T\sbra{\m{u}}} + c_j, \end{split}}
where $c_j$ is a constant independent of $\m{u}$. As a result, the optimization problem at the $\sbra{j+1}$th iteration becomes
\equ{\begin{split}
&\min_{\m{Y}\in\cS,\m{u}} \tr\mbra{\sbra{T\sbra{\m{u}_j}+\epsilon\m{I}}^{-1}T\sbra{\m{u}}} + \tr\sbra{\m{Y}^HT\sbra{\m{u}}^{-1}\m{Y}},\\
&\st T\sbra{\m{u}}\geq\m{0}. \end{split} \label{formu:problem_j}}
Since $\ln\abs{T\sbra{\m{u}}+\epsilon\m{I}}$ is strictly concave in $\m{u}$, at each iteration its value decreases by an amount greater than the decrease of its tangent plane. It follows that by iteratively solving \eqref{formu:problem_j} the objective function in (\ref{formu:problem}) monotonically decreases and converges to a local minimum.

\subsection{Interpretation as RAM }\label{sec:WAN_SDP}
To interpret the optimization problem in (\ref{formu:problem_j}), let us define a weighted continuous dictionary
\equ{\cA^w\triangleq\lbra{\m{a}^w\sbra{f}= w\sbra{f}\m{a}\sbra{f}:\; f\in\bT}}
w.r.t. the original continuous dictionary $\lbra{\m{a}\sbra{f}:\; f\in\bT}$, where $w\sbra{f}\geq 0$ is a weighting function. For $\m{Y}\in\bC^{N\times L}$, we define its weighted atomic norm w.r.t. $\cA^w$ as its atomic norm induced by $\cA^w$:
\equ{\begin{split}
\norm{\m{Y}}_{\cA^w}
&\triangleq\inf_{f_k,\m{s}^w_k}\lbra{\sum_k \twon{\m{s}^w_k}: \m{Y} = \sum_k \m{a}^w\sbra{f_k}\m{s}^w_k}\\
&=\inf_{f_k,\m{s}_k}\lbra{\sum_k \frac{\twon{\m{s}_k}}{w\sbra{f_k}}: \m{Y} = \sum_k \m{a}\sbra{f_k}\m{s}_k }. \end{split}}
According to the definition above, $w\sbra{f}$ specifies preference of the atoms $\lbra{\m{a}\sbra{f}}$. To be specific, an atom $\m{a}\sbra{f_0}$, $f_0\in\bT$, is more likely selected if $w\sbra{f_0}$ is larger. Moreover, the atomic norm is a special case of the weighted atomic norm with a constant weighting function (i.e., without any preference). Similar to the atomic norm, the proposed weighted atomic norm also admits a semidefinite formulation when assigned an appropriate weighting function, which is stated in the following theorem.

\begin{thm} Suppose that $w\sbra{f}= \frac{1}{\sqrt{\m{a}\sbra{f}^H\m{W}\m{a}\sbra{f}}}$ with $\m{W}\in\bC^{N\times N}$. Then,
\equ{\begin{split}
\norm{\m{Y}}_{\cA^w}=&\min_{\m{u}} \frac{\sqrt{N}}{2}\tr\sbra{\m{W}T\sbra{\m{u}}} + \frac{1}{2\sqrt{N}}\tr\sbra{\m{Y}^HT\sbra{\m{u}}^{-1}\m{Y}},\\
&\st T\sbra{\m{u}}\geq\m{0}. \end{split}} \label{thm:weightAN}
\end{thm}
\begin{proof} See Appendix \ref{sec:Append_weightAN}.
\end{proof}

Let $\m{W}_j=\frac{1}{N}\sbra{T\sbra{\m{u}_j}+\epsilon\m{I}}^{-1}$ and $w_j\sbra{f} = \frac{1}{\sqrt{\m{a}\sbra{f}^H \m{W}_j \m{a}\sbra{f}}}$. It follows from Theorem \ref{thm:weightAN} that the optimization problem in (\ref{formu:problem_j}) can be exactly written as the following weighted atomic norm minimization problem:
\equ{\begin{split}
\min_{\m{Y}\in\cS} \norm{\m{Y}}_{\cA^{w_j}}. \end{split} \label{formu:WAN_j}}
As a result, the proposed iterative algorithm can be interpreted as reweighted atomic-norm minimization (RAM), where the weighting function is updated based on the latest solution of $\m{u}$. If we let $w_0(f)$ be a constant function or equivalently, $\m{u}_0=\m{0}$, such that no preference of the atoms is specified at the first iteration, then the first iteration coincides with ANM. From the second iteration on, the preference is defined by the weighting function $w_j\sbra{f}$ given above. Note that $w_j^2(f)$ is in fact the power spectrum of Capon's beamforming provided that $T\sbra{\m{u}_j}$ is interpreted as the noiseless data covariance following from Remark \ref{rem:datacov} and $\epsilon$ as the noise variance. Therefore, the reweighting strategy makes the frequencies around those produced by the current iteration preferable at the next iteration and thus enhances sparsity. At the same time, the preference results in finer details of the frequency spectrum in those areas and therefore enhances resolution. Empirical evidences will be provided in Section \ref{sec:simulation}.

\section{Computationally Efficient Implementations} \label{sec:implementation}

\subsection{Optimization Using Standard SDP Solver}
At each iteration of RAM, we need to solve the SDP in (\ref{formu:problem_j}) as follows:
\equ{\begin{split}
&\min_{\m{Y}\in\cS,\m{u},\m{X}} \tr\sbra{\m{W}T\sbra{\m{u}}} +\tr\sbra{\m{X}}, \\
&\st \begin{bmatrix}\m{X} & \m{Y}^H \\ \m{Y} & T\sbra{\m{u}}\end{bmatrix} \geq\m{0},
\end{split} \label{formu:weightedSDP}}
where $\m{W}=\sbra{T\sbra{\m{u}_j}+\epsilon\m{I}}^{-1}$.
Its dual problem is given as follows by a standard Lagrangian analysis (see Appendix \ref{sec:append_duality}):
\equ{\begin{split}
&\min_{\m{V},\m{Z}} 2\eta\frobn{\m{V}_{\m{\Omega}}} + 2 \Re\tr\sbra{\m{Y}_{\m{\Omega}}^{oH}\m{V}_{\m{\Omega}} },\\
& \st \begin{bmatrix}\m{I} & \m{V}^H \\ \m{V} & \m{Z}\end{bmatrix} \geq\m{0}, \; \m{V}_{\overline{\m{\Omega}}} = \m{0}, \\
& \phantom{\st} \sum_{n=1}^{N-j}Z_{n,n+j} = \sum_{n=1}^{N-j}W_{n,n+j}, \;\; j=0,\dots,N-1,\\
\end{split} \label{formu:duality}}
where $\Re$ takes the real part of the argument and $Z_{n,j}$ denotes the $(n,j)$th entry of $\m{Z}$. We empirically find that the dual problem (\ref{formu:duality}) can be solved more efficiently than the primal problem (\ref{formu:weightedSDP}) with a standard SDP solver SDPT3 \cite{toh1999sdpt3}. Note that the optimizer to (\ref{formu:weightedSDP}) is given for free via duality when we solve (\ref{formu:duality}). As a result, the reweighted algorithm can be iteratively implemented.


\subsection{A First-order Algorithm via ADMM}
A reasonably fast approach for ANM is based on ADMM \cite{boyd2011distributed,bhaskar2013atomic,yang2015gridless}, which is a first-order algorithm and guarantees global optimality. To derive the ADMM algorithm, we reformulate the SDP in (\ref{formu:weightedSDP}) as follows:
\equ{\begin{split}
&\min_{\m{u},\m{X},\m{Y}\in\cS, \m{\cQ}\geq\m{0}} \tr\sbra{\m{W}T\sbra{\m{u}}} +\tr\sbra{\m{X}}, \\
&\st \m{\cQ} = \begin{bmatrix}\m{X} & \m{Y}^H \\ \m{Y} & T\sbra{\m{u}}\end{bmatrix},
\end{split} \label{formu:weightedSDP2}}
which is very similar to the SDP solved in \cite{yang2015gridless}.
Then we can write the augmented Lagrangian function and iteratively update $\sbra{\m{u},\m{X},\m{Y}}$, $\m{\cQ}$ and the Lagrangian multiplier in closed-form expressions until convergence. We omit the details since all the formulae and derivations are similar to those in \cite{yang2015gridless}, to which interested readers are referred. We mention that an eigen-decomposition of a matrix of order $N+L$ (the order of $\m{\cQ}$) is required at each iteration. Note that the ADMM converges slowly to an extremely accurate solution while moderate accuracy is typically sufficient in practical applications \cite{boyd2011distributed}.


\subsection{Dimension Reduction for Large $L$}
The number of measurement vectors $L$ can be large, possibly with $L\gg M$, in DOA estimation, which increases considerably the computational workload. We provide the following result to reduce this number from $L$ to $\rank\sbra{\m{Y}_{\m{\Omega}}^o}\leq \min(L,M)$.

\begin{prop} Let $r=\rank\sbra{\m{Y}_{\m{\Omega}}^o}\leq \min(L,M)$. Find a unitary matrix $\m{Q}\in\bC^{L\times L}$ (for example, by QR decomposition) satisfying that $\m{Y}_{\m{\Omega}}^{o}\m{Q}= \begin{bmatrix}\m{Y}_{\m{\Omega}}^{o}\m{Q}_1 & \m{0}\end{bmatrix}$,
where $\m{Q}_1\in\bC^{L\times r}$. If we make the substitutions $\m{Y}_{\m{\Omega}}^{o}\rightarrow\m{Y}_{\m{\Omega}}^{o}\m{Q}_1$ and $\m{Y}\in\bC^{N\times L}\rightarrow\m{Z}\in\bC^{N\times r}$ in (\ref{formu:problem_j}) and denote by $\sbra{\m{Z}^*,\m{u}^*}$ the optimizer to the resulting optimization problem, then the optimizer to (\ref{formu:problem_j}) is given by $\sbra{\m{Z}^*\m{Q}_1^H,\m{u}^*}$. The same result holds for the nonconvex optimization problem in (\ref{formu:criterion3}). \label{prop:dimreduce}
\end{prop}
\begin{proof} See Appendix \ref{sec:Append_dimreduce}.
\end{proof}

\begin{rem} If only the frequencies are of interest, e.g., in DOA estimation, we can replace $\m{Y}_{\m{\Omega}}^{o}\m{Q}_1$ by $\sbra{\m{Y}_{\m{\Omega}}^{o}\m{Y}_{\m{\Omega}}^{oH}}^{\frac{1}{2}}$ (in fact, any matrix $\widetilde{\m{Y}}$ satisfying that $\widetilde{\m{Y}}\widetilde{\m{Y}}^H=\m{Y}_{\m{\Omega}}^{o}\m{Y}_{\m{\Omega}}^{oH}$) to obtain the same solution of $\m{u}$. \label{rem:sameusolution}
\end{rem}

\begin{rem} With a similar proof, the dimension reduction technique in Proposition \ref{prop:dimreduce} can be extended to a more general linear model with observations expressed by $\m{\Phi}\m{Y}^o$ + noise, where $\m{\Phi}$ denotes a sensing matrix. Also, it can be applied to conventional discrete dictionary models, in which, for example, the $\ell_{2,p}$ norm, $0\leq p\leq1$, needs to be optimized. Note that the dimension reduction approach introduced here is different from that in \cite{malioutov2005sparse}. In particular, the approach in \cite{malioutov2005sparse} requires the knowledge of the model order $K$, which is unknown in practical scenarios, and gives an optimization problem which is an approximation of the original one. In contrast, our approach does not need $K$ but produces a dimension-reduced, equivalent problem.
\end{rem}

By Proposition \ref{prop:dimreduce}, when $L>M$ we can reduce the order of the PSD matrix in (\ref{formu:weightedSDP}) [and (\ref{formu:duality}), (\ref{formu:weightedSDP2})] from $N+L$ to $N+r\leq N+M$. Therefore, the resulting problem dimension depends only on $M$ and $N$. Both the SDPT3 and ADMM implementations of RAM above can be reasonably fast when $M$ and $N$ are small though they may not possess good scalability, especially for SDPT3. An application at hand is DOA estimation in which the array size $M$ and aperture $N$ are typically small (on the order of 10) though $L$ can be a few hundred or even greater. Note also that the dimension reduction technique takes $O\sbra{M^2L}$ flops in DOA estimation according to Remark \ref{rem:sameusolution}. Extensive numerical simulations will be provided in Section \ref{sec:simulation} to demonstrate usefulness of our method.



\subsection{Remarks on Algorithm Implementation}
According to the discussions in Section \ref{sec:WAN_SDP}, we can always start with $\m{u}_0=\m{0}$ and the first iteration coincides with ANM. When $L$ is large, we can also implement a weighting function in the first iteration inspired by Capon's beamforming for faster convergence (see an example in Section \ref{sec:simulation_DOA}). Moreover, we gradually decrease $\epsilon$ during the algorithm and define the weighting function using the latest solution for avoiding local minima (note that the first iteration with $\m{u}_0=\m{0}$ essentially corresponds to $\epsilon=+\infty$ following from Theorem \ref{thm:epsilontoinf}). In fact, this is like an aggressive continuation strategy in which we attempt to solve the nonconvex optimization problem in (\ref{formu:criterion3}) at decreasing values of $\epsilon$. The convergence of the reweighted algorithm is retained if we fix $\epsilon$ when it is sufficiently small. In the ADMM implementation, we can further accelerate the algorithm by adopting loose convergence criteria in the first few iterations of RAM. Finally, note that we need to trade off the algorithm performance for computational time by keeping the number of iterations of RAM being small.

\section{Numerical Simulations} \label{sec:simulation}
\subsection{Implementation Details of RAM}
In our implementation of RAM, we first scale the measurements and the noise such that $\frobn{\m{Y}_{\m{\Omega}}}^2=M$ (the noise energy becomes $\eta'^2=\frac{M\eta^2}{\frobn{\m{Y}}^2}$) and compensate the recovery afterwards. We start with $\m{u}_0=\m{0}$ and $\epsilon=1$ as default. We halve $\epsilon$ when beginning a new iteration until $\epsilon=\frac{1}{2^{10}}$ or $\epsilon<\frac{\eta'^2}{10}$. When $\eta=0$ we terminate RAM if the relative change (in the Frobenius norm) of the solution $\m{Y}^*$ at two consecutive iterations is less than $10^{-6}$ or the maximum number of iterations, set to 20, is reached. All simulations are carried out in Matlab v.8.1.0 on a PC with a Windows 7 system and a 3.4 GHz CPU.

\subsection{An Illustrative Example}
We provide a simple example in this subsection to illustrate the iterative process of RAM. In particular, we consider a sparse frequency spectrum consisting of $K=5$ spikes located at 0.1, 0.108, 0.125, 0.2 and 0.5. We randomly generate the complex amplitudes and randomly select $M=30$ samples among $N=64$ consecutive time-space measurements, with $L=1$. Then, we run the RAM algorithm to reconstruct the frequency spectrum from the samples. Note that the first three frequencies are mutually separated by only about $\frac{0.51}{N}$ and $\frac{1.09}{N}$. Implemented with SDPT3 the RAM algorithm converges in four iterations. We plot the simulation results in Fig. \ref{Fig:illusexample}, where the first subfigure presents variation of the eigenvalues of $T\sbra{\m{u}^*}$ during the iterations, the second row presents the recovered spectra of the first three iterations, and the last row plots the weighting functions used in the first three iterations. Note that the first iteration, which exploits a constant weighting function and coincides with ANM, can detect a rough area where the first three spikes are located but cannot accurately determine their locations and number. In the second iteration, a weighting function is implemented based on the previous estimate to provide preference of the frequencies around those produced in the first iteration. As a result, the third spike is identified while the first two are still not. Following from the same reweighting process, all the frequencies are correctly determined in the next iteration and the algorithm converges after that. It is worth noting that $T\sbra{\m{u}^*}$ becomes rank-5 and the remaining eigenvalues become zero (within numerical precision) after three iterations, where $\epsilon=0.25$. Finally, we report that the relative mean squared error (MSE) of signal recovery improves during the iterations from $1.28\times10^{-4}$ to $2.01\times 10^{-7}$, $3.16\times10^{-19}$ and $3.01\times10^{-21}$. Each iteration takes about 1.7s.

\begin{figure*}
\centering
  \includegraphics[width=6in]{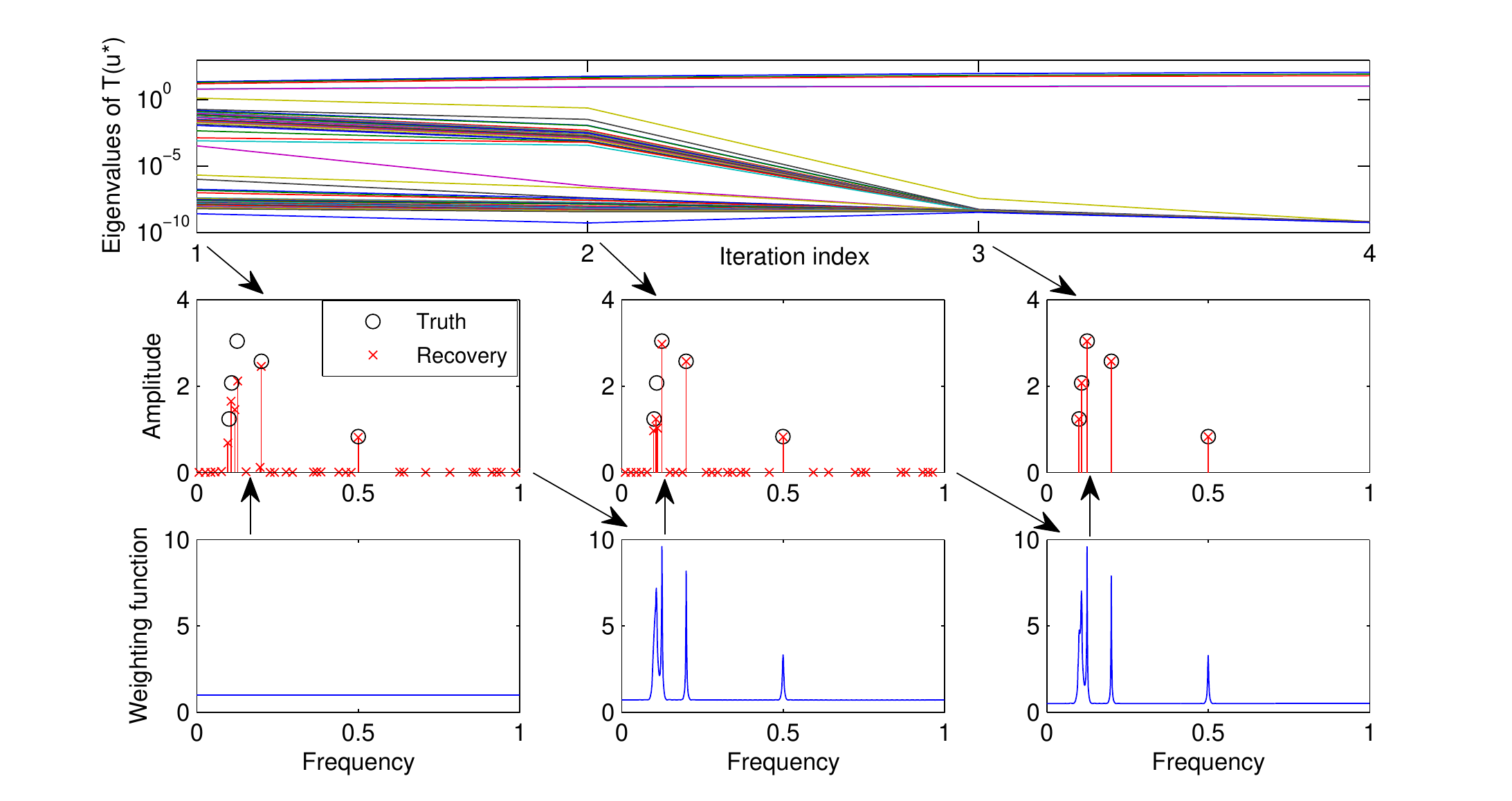}
\centering
\caption{An illustrative example of RAM. Some settings include $N=64$, $M=30$, $K=5$ with frequencies located at 0.1, 0.108, 0.125, 0.2 and 0.5. The first row presents variation of eigenvalues of $T\sbra{\m{u}^*}$ w.r.t. the iteration index. Only 5 eigenvalues remain nonzero (within numerical precision) after 3 iterations. The second row presents the recovered spectra of the first 3 iterations. The last row plots the weighting functions used in the first 3 iterations to produce the spectra. Note that the first iteration coincides with the ANM in which a constant weighting function is used.} \label{Fig:illusexample}
\end{figure*}


\subsection{Sparsity-Separation Phase Transition}
In this subsection, we study the success rate of RAM in signal and frequency recovery compared to ANM. In particular, we fix $N=64$ and $M=30$ with the sampling index set $\m{\Omega}$ being generated uniformly at random. We vary the duo $\sbra{K,\Delta_f}$ and for each combination we randomly generate $K$ frequencies such that they are mutually separated by at least $\Delta_f$. We randomly generate the amplitudes $\lbra{s_{kt}}$ independently and identically from a standard complex normal distribution. After obtaining the noiseless samples, we carry out signal reconstruction and frequency recovery using ANM and RAM, both implemented by SDPT3. The recovery is called successful if both the relative MSE of signal recovery and the MSE of frequency recovery are less than $10^{-12}$. For each combination $\sbra{K,\Delta_f}$, the success rate is measured over 20 Monte Carlo runs.

We plot the success rates of ANM and RAM with $L=1$ in Fig. \ref{Fig:phasetrans_1}, where it is shown that successful recovery can be obtained with more ease in the case of a smaller $K$ and a larger frequency separation $\Delta_f$, leading to a phase transition in the sparsity-separation domain. By comparing the two images, we see that RAM significantly enlarges the success phase and enhances sparsity and resolution. It is worth noting that the phase transitions of both ANM and RAM are not sharp. One reason is that, a set of \emph{well separated} frequencies can be possibly generated at a small value of $\Delta_f$ while we only control that the frequencies are separated by \emph{at least} $\Delta_f$. It is also observed that RAM tends to converge in less iterations with a smaller $K$ and a larger $\Delta_f$.

\begin{figure}
\centering
\includegraphics[width=1.65in]{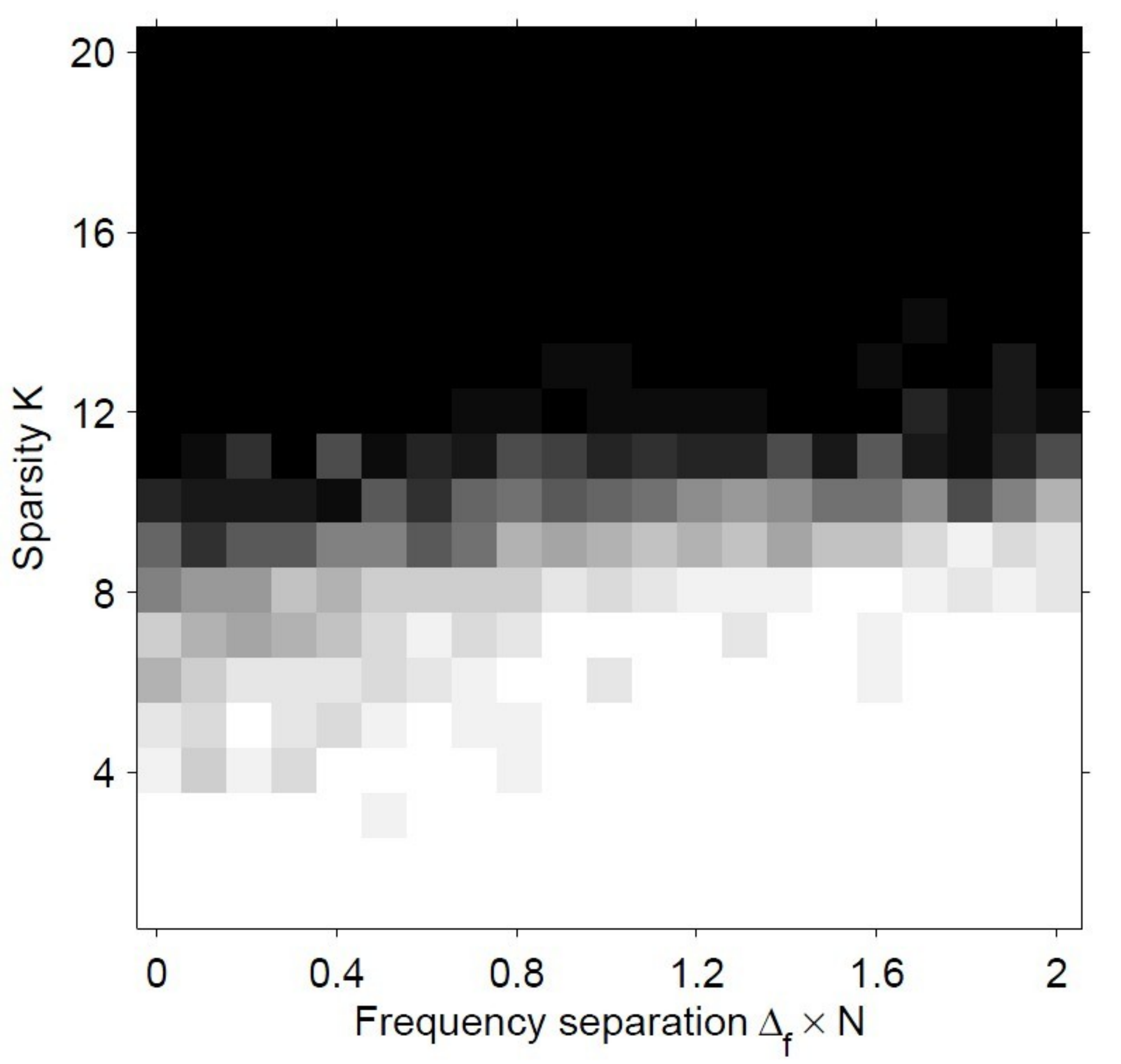} %
\includegraphics[width=1.63in]{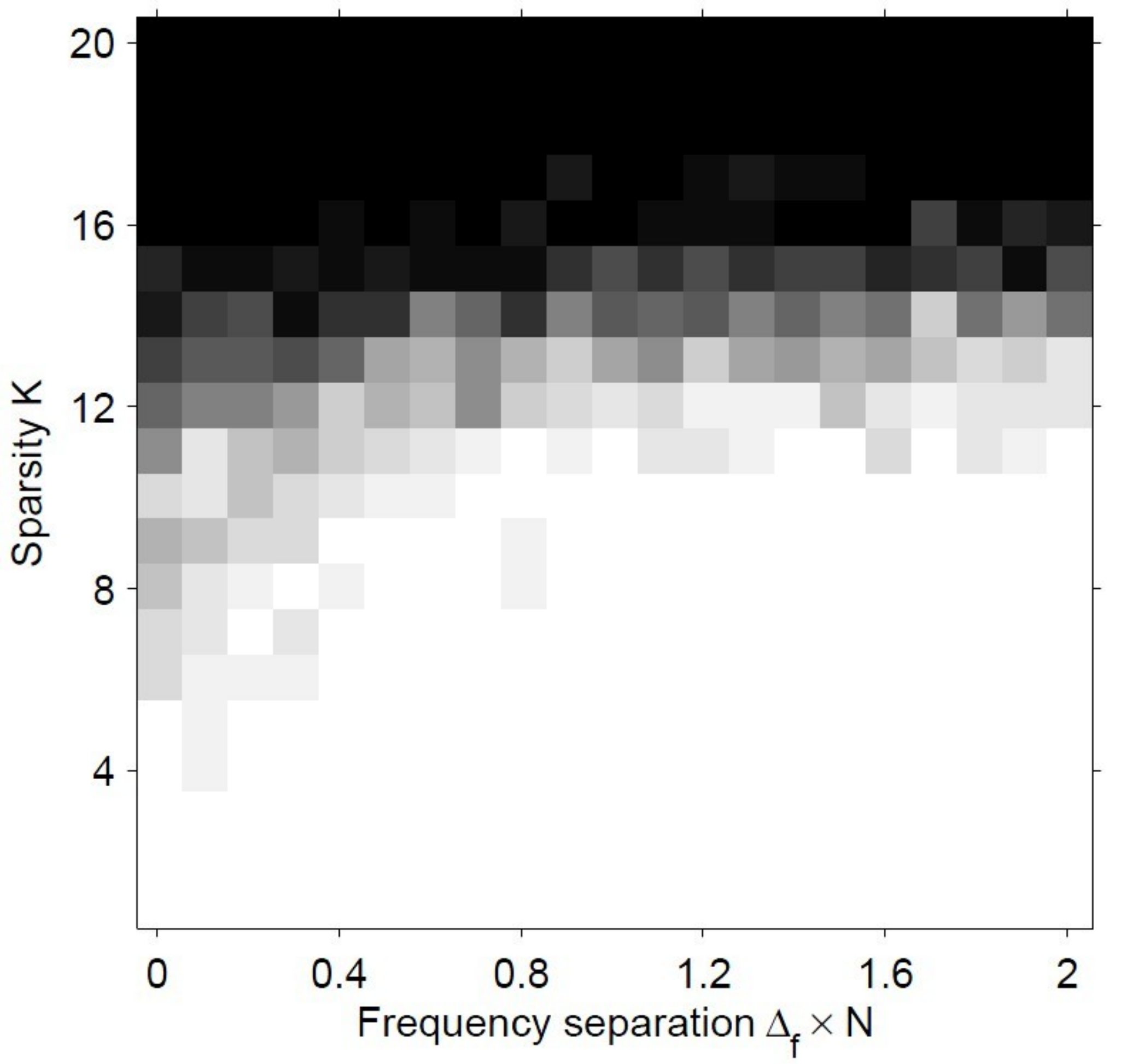}
\caption{Sparsity-separation phase transition of ANM (left) and RAM (right) with $L=1$, $N=64$ and $M=30$. The grayscale images present the success rates, where white and black indicate complete success and complete failure, respectively.} \label{Fig:phasetrans_1}
\end{figure}

We also consider the MMV case with $L=5$. The success rates of ANM and RAM are presented in Fig. \ref{Fig:phasetrans_5}. Again, remarkable improvement is obtained by the proposed RAM compared to ANM. In fact, we did not find a single failure in our simulation whenever $K\leq20$ and $\Delta_f\geq\frac{0.3}{N}$. By comparing the results in Figs. \ref{Fig:phasetrans_1} and \ref{Fig:phasetrans_5}, it can be observed that improved signal and frequency recovery performance can be obtained by increasing $L$, as reported in \cite{yang2014continuous,yang2014exact}.

\begin{figure}
\centering
\includegraphics[width=1.65in]{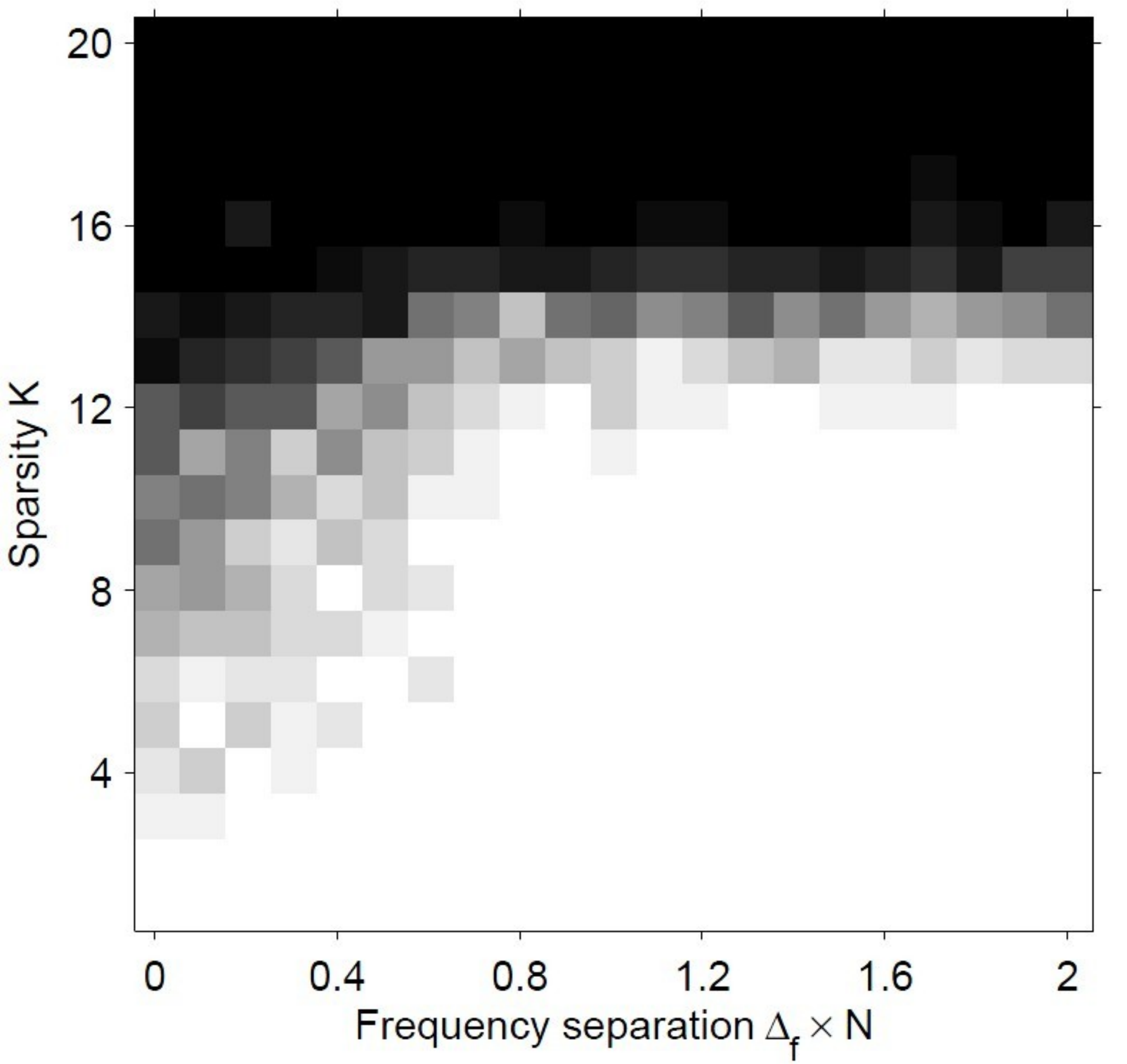} %
\includegraphics[width=1.63in]{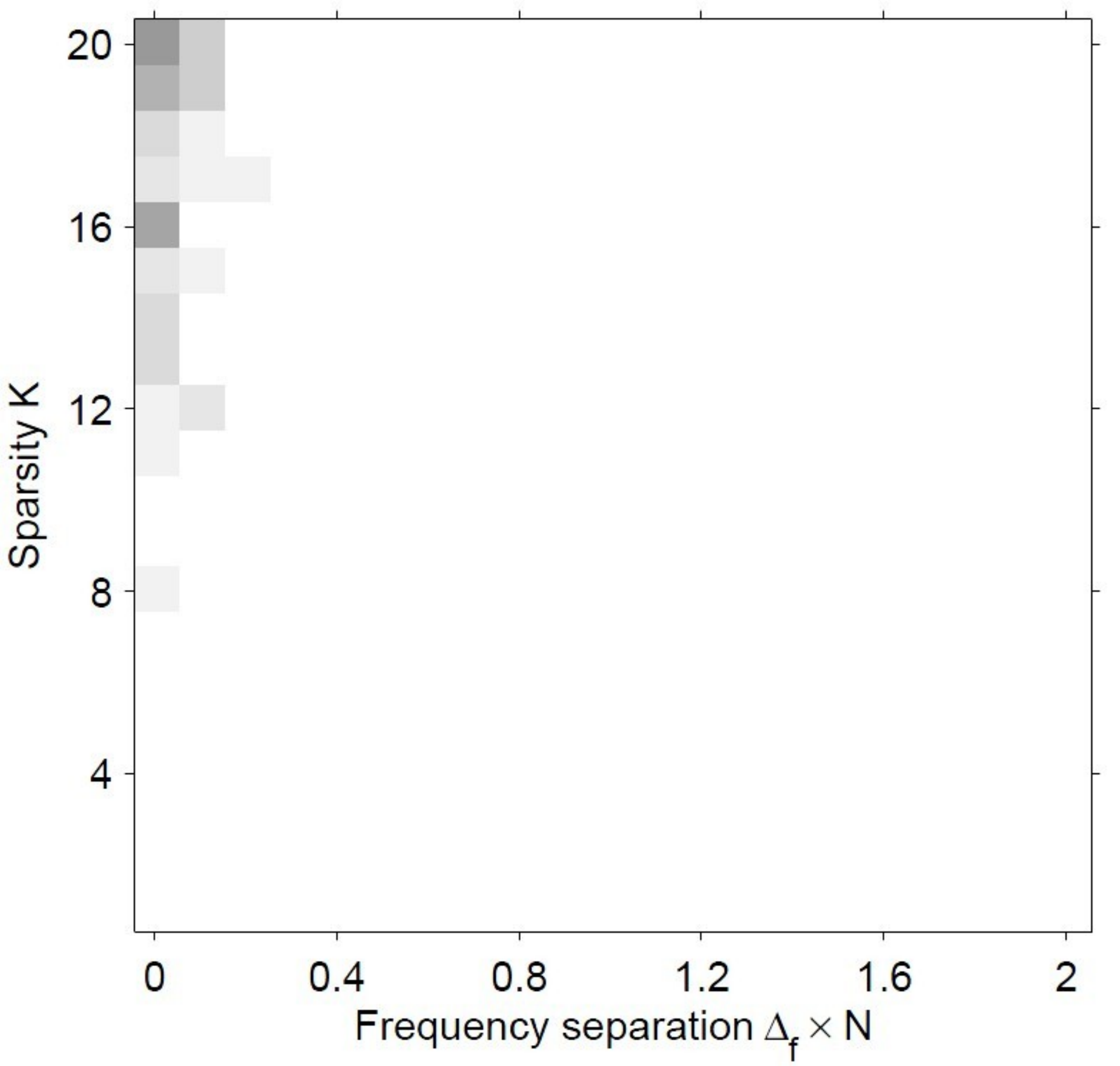}
\centering
\caption{Sparsity-separation phase transition of ANM (left) and RAM (right) with $L=5$, $N=64$ and $M=30$. The grayscale images present the success rates, where white and black indicate complete success and complete failure, respectively.} \label{Fig:phasetrans_5}
\end{figure}

\subsection{Application to DOA Estimation} \label{sec:simulation_DOA}

We apply the proposed RAM method to DOA estimation. In particular, we consider a 10-element SLA that is obtained from a virtual 20-element ULA, in which adjacent antennas are spaced by half a wavelength, by retaining the antennas indexed by $\m{\Omega}=\lbra{1,2,5,6,8,12,15,17,19,20}$. Hence, we have that $N=20$ and $M=10$. Consider that $K=4$ narrowband sources impinge on the array from directions corresponding to frequencies $0.1$, $0.11$, $0.2$ and $0.5$, and powers $10$, $10$, $3$ and $1$, respectively. Therefore, it is challenging to separate the first two sources which are separated by only $\frac{0.2}{N}$. We consider both cases of uncorrelated and correlated sources. In the latter case, sources 1 and 3 are set to be coherent (completely correlated). Assume that $L=200$ data snapshots are collected which are corrupted by i.i.d. Gaussian noise of unit variance. In our simulation, ANM and RAM are implemented using both SDPT3 and ADMM and based on the proposed dimension reduction technique. A nontrivial weighting function is implemented in the first iteration of RAM with $\m{W}=\m{\Gamma}_{\m{\Omega}}^T \sbra{\frac{1}{L}\m{Y}_{\m{\Omega}}^o\m{Y}_{\m{\Omega}}^{oH}+\epsilon\m{I}}^{-1} \m{\Gamma}_{\m{\Omega}}$, where $\m{\Gamma}_{\m{\Omega}}\in\lbra{0,1}^{M\times N}$ has 1 in the $j$th
row only at the $\Omega_j$th position. The weighting function corresponds to Capon's beamforming with $\frac{1}{L}\m{Y}_{\m{\Omega}}^o\m{Y}_{\m{\Omega}}^{oH}$ being the sample covariance and $\epsilon$ a regularization parameter. We terminate RAM within maximally 10 iterations. We consider MUSIC and ANM for comparison. Assume that the noise variance $\sigma^2=1$ is given for ANM and RAM and the source number $K$ is provided for MUSIC. We set $\eta^2=\sbra{ML+2\sqrt{ML}}\sigma^2$ (mean + twice standard deviation) to upper bound the noise energy with high probability in ANM and RAM.

\begin{figure}
\centering
\includegraphics[width=1.65in]{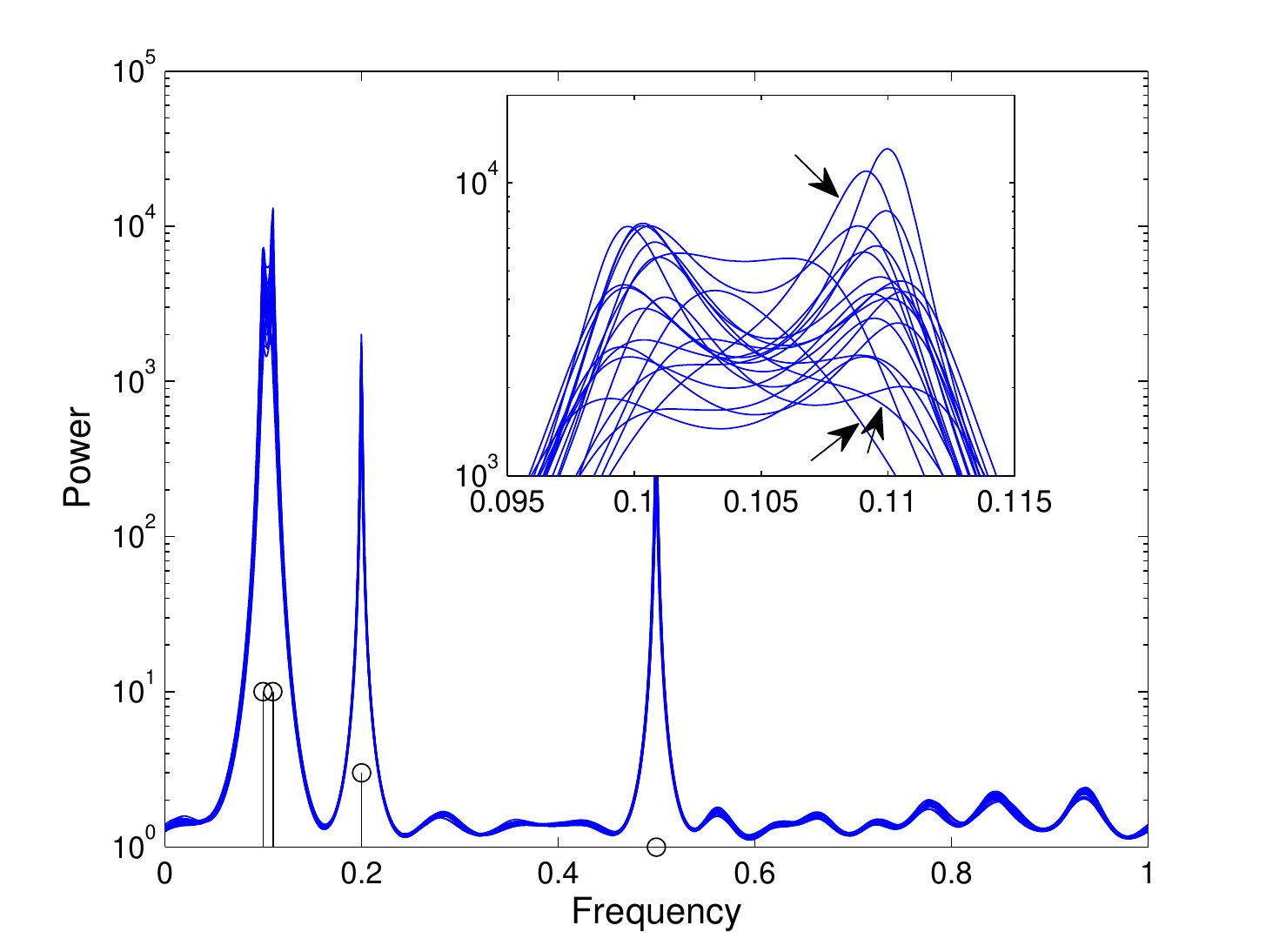} %
\includegraphics[width=1.63in]{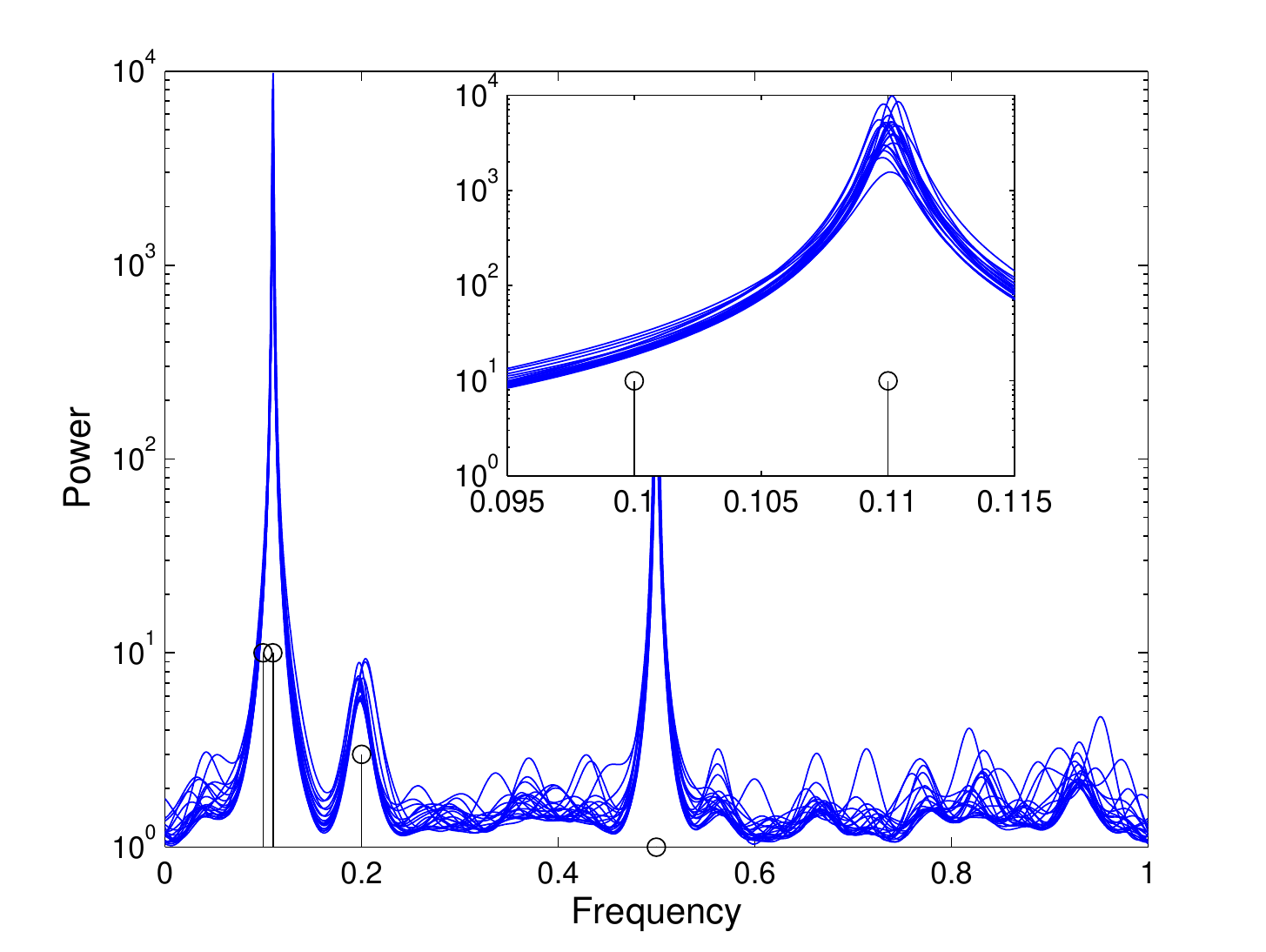}
\includegraphics[width=1.65in]{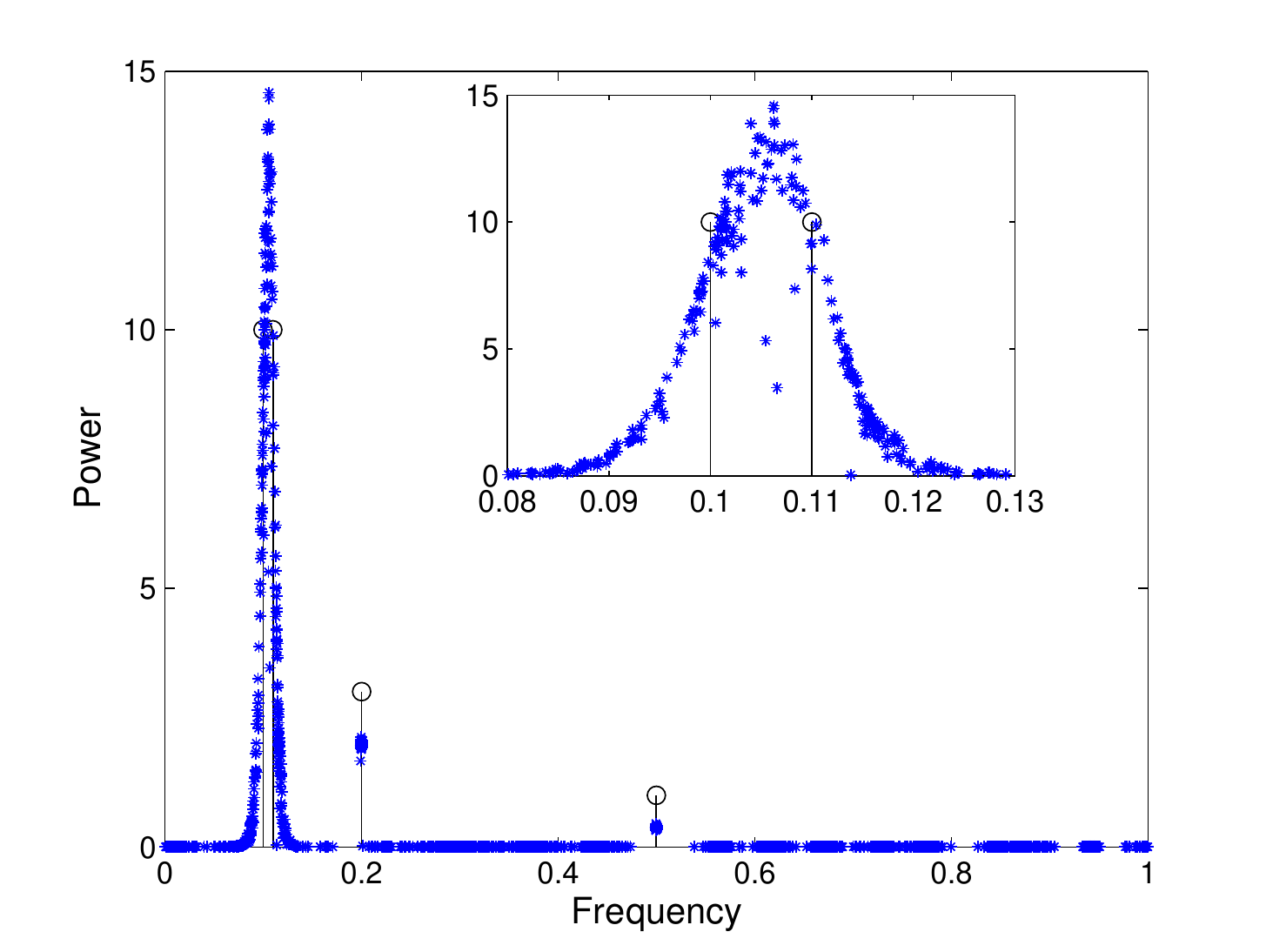} %
\includegraphics[width=1.63in]{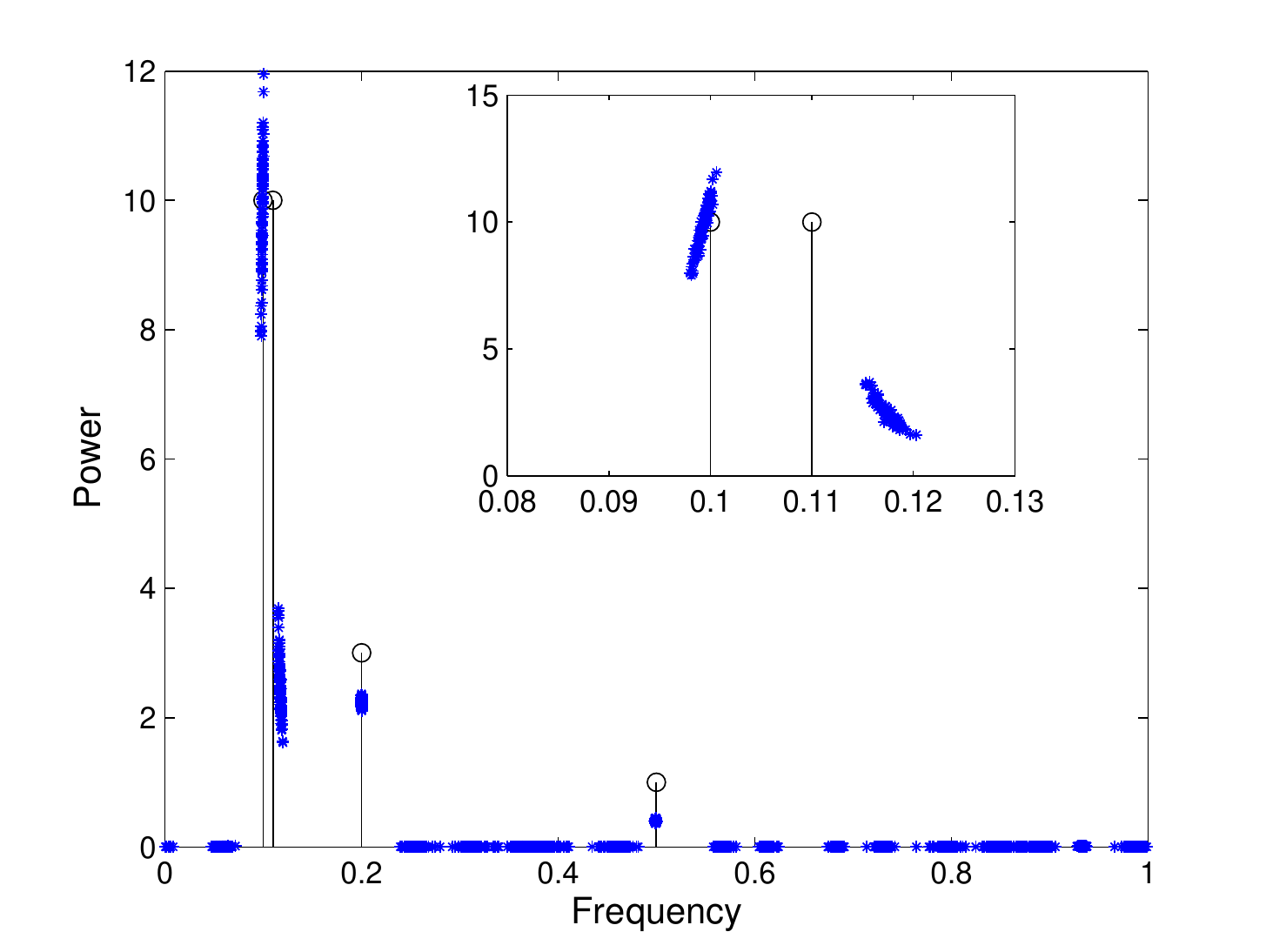}
\includegraphics[width=1.65in]{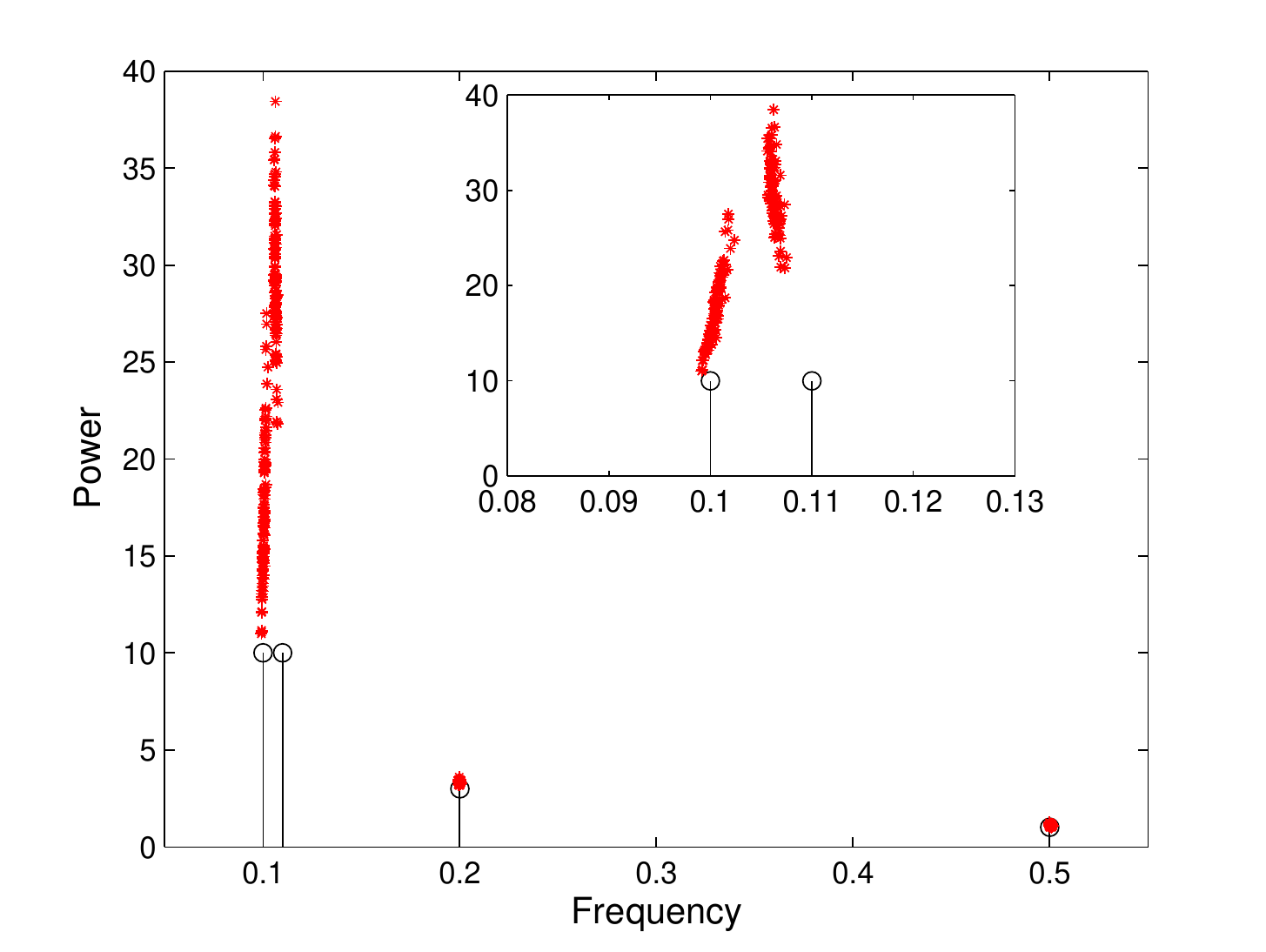} %
\includegraphics[width=1.63in]{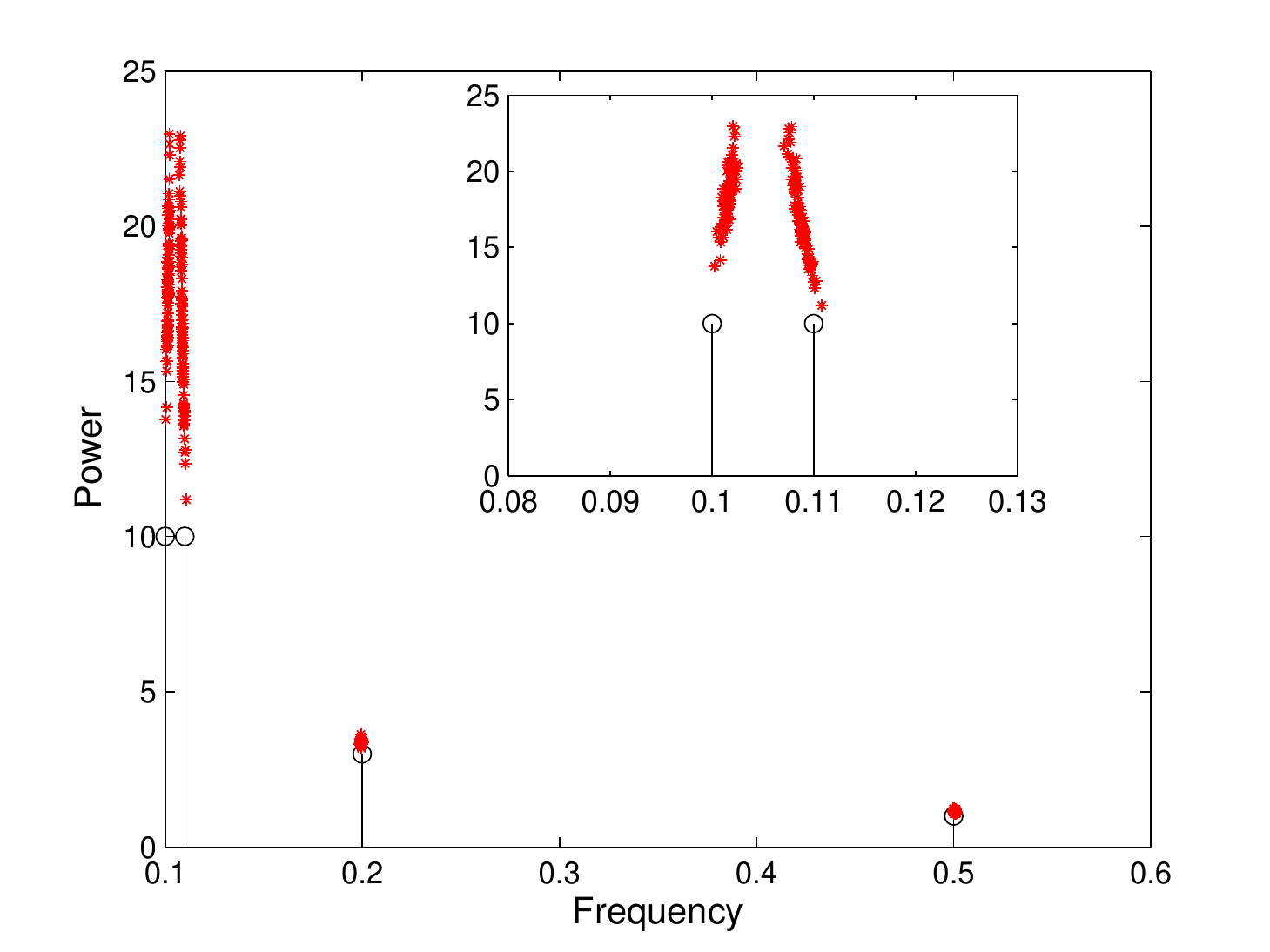}
\centering
\caption{Frequency spectra of MUSIC (top), ANM (middle) and RAM (bottom) with uncorrelated (left) and correlated (right) sources in 100 Monte Carlo runs. Sources 1 and 3 are coherent in the case of correlated sources. The area around the first two sources are zoomed in in each subfigure. Only results of the first 20 runs are presented for MUSIC for clearer illustration.} \label{Fig:noisy_noiselevel}
\end{figure}

Our simulation results of 100 Monte Carlo runs are presented in Fig. \ref{Fig:noisy_noiselevel} (only the first 20 runs are presented for MUSIC for better illustration). In the absence of source correlations, MUSIC has satisfactory performance in most scenarios. However, its power spectrum exhibits only a single peak around the first two sources (i.e., the two sources cannot be separated) in at least 3 out of the first 20 runs (indicated by the arrows). Moreover, MUSIC is sensitive to source correlations and cannot detect source 1 when it is coherent with source 3. ANM cannot separate the first two sources in the uncorrelated source case and always produces many spurious sources. In contrast, the proposed RAM always correctly detects 4 sources near the true locations, demonstrating its capabilities in enhancing sparsity and resolution. It is interesting to note that the first two sources tend to merge together in RAM. This is reasonable since in the case of heavy noise it is even possible that we can only detect a single source around the first two and 3 sources in total. Note also that the results of ANM and RAM presented in Fig. \ref{Fig:noisy_noiselevel} are produced by the ADMM implementations, while those by SDPT3 are very similar and omitted. In computational time, the SDPT3 versions of ANM and RAM take $0.87$s and $7.31$s on average, respectively, while these numbers are decreased to $0.20$s and $2.65$s for the ADMM ones.



\section{Conclusion} \label{sec:conclusion}
In this paper, we studied the signal and frequency recovery problem in CCS. Motivated by its connection to the topic of LRMR, we proposed reweighted atomic-norm minimization (RAM) for enhancing sparsity and resolution compared to currently prominent atomic norm minimization (ANM) and validated its advantageous performance via numerical simulations. As a byproduct, we have established a framework for applying LRMR techniques to CCS. In future studies, we may try other methods for matrix rank minimization in the literature and propose more computationally efficient algorithms for CCS. While LRMR represents a 2D counterpart of sparse recovery in discrete CS, this work sheds light on connections between discrete CS, continuous CS and LRMR.

\appendix

\subsection{Proof of Theorem \ref{thm:epsilontoinf}} \label{sec:append:epsilontoinf}
Note that
\equ{\begin{split}
&\cM^{\epsilon}\sbra{\m{Y}}-N\ln\epsilon \\
&=\min_{\m{u}} \ln\abs{\epsilon^{-1}T\sbra{\m{u}} +\m{I} } + \tr\sbra{\m{Y}^HT\sbra{\m{u}}^{-1} \m{Y}},\\
&\phantom{=}\st T\sbra{\m{u}}\geq\m{0}. \end{split}\label{formu:resid}}
Let
\equ{\begin{split}\m{u}^*=
&\arg\min_{\m{u}} \tr\sbra{T\sbra{\m{u}}} + \tr\sbra{\m{Y}^HT\sbra{\m{u}}^{-1} \m{Y}},\\
&\st T\sbra{\m{u}}\geq\m{0}.\end{split}}
Then, according to (\ref{formu:AN_SDP}) we have that
\equ{\tr\sbra{T\sbra{\m{u}^*}} + \tr\sbra{\m{Y}^HT\sbra{\m{u}^*}^{-1} \m{Y}}=2\sqrt{N}\atomn{\m{Y}}.}
Consider the value of the objective function in (\ref{formu:resid}) at $\m{u}=\epsilon^{\frac{1}{2}}\m{u}^*$. It holds that
\equ{\begin{split}
&\cM^{\epsilon}\sbra{\m{Y}}-N\ln\epsilon\\
&\leq \ln\abs{\epsilon^{-\frac{1}{2}}T\sbra{\m{u}^*} +\m{I} } + \tr\sbra{\m{Y}^HT\sbra{\m{u}^*}^{-1} \m{Y}}\epsilon^{-\frac{1}{2}} \\
&= \tr\sbra{T\sbra{\m{u}^*}} \epsilon^{-\frac{1}{2}} + o\sbra{\epsilon^{-\frac{1}{2}}} + \tr\sbra{\m{Y}^HT\sbra{\m{u}^*}^{-1} \m{Y}}\epsilon^{-\frac{1}{2}}\\
&= 2\sqrt{N}\atomn{\m{Y}}\epsilon^{-\frac{1}{2}} + o\sbra{\epsilon^{-\frac{1}{2}}}.
\end{split} \label{formu:residleq}}
On the other hand, we denote by $\m{u}_{\epsilon}^*$ the optimizer to the optimization problem in (\ref{formu:resid}). We first argue that $\tr\sbra{T\sbra{\m{u}_{\epsilon}^*}}\epsilon^{-1}=o\sbra{1}$. Otherwise, by (\ref{formu:resid}) $\cM^{\epsilon}\sbra{\m{Y}}-N\ln\epsilon\geq \ln\abs{\epsilon^{-1}T\sbra{\m{u}_{\epsilon}^*} +\m{I} }$ is not an infinitesimal, contradicting (\ref{formu:residleq}). Therefore,
\equ{\begin{split}
&\cM^{\epsilon}\sbra{\m{Y}}-N\ln\epsilon\\
&= \ln\abs{\epsilon^{-1}T\sbra{\m{u}_{\epsilon}^*} +\m{I} } + \tr\sbra{\m{Y}^HT\sbra{\m{u}_{\epsilon}^*}^{-1} \m{Y}}\\
&= \tr\sbra{T\sbra{\m{u}_{\epsilon}^*}} \epsilon^{-1} + o\sbra{\tr\sbra{T\sbra{\m{u}_{\epsilon}^*}}\epsilon^{-1}}\\
&\quad+ \tr\sbra{\m{Y}^HT\sbra{\m{u}_{\epsilon}^*}^{-1} \m{Y}}\\
&= \mbra{\tr\sbra{T\sbra{\m{u}_{\epsilon}^*}} + \tr\sbra{\epsilon^{\frac{1}{2}}\m{Y}^HT\sbra{\m{u}_{\epsilon}^*}^{-1} \m{Y}\epsilon^{\frac{1}{2}}}}\epsilon^{-1}\\
&\quad+ o\sbra{\tr\sbra{T\sbra{\m{u}_{\epsilon}^*}}\epsilon^{-1}}\\
&\geq 2\sqrt{N}\atomn{\epsilon^{\frac{1}{2}}\m{Y}}\epsilon^{-1} + o\sbra{\tr\sbra{T\sbra{\m{u}_{\epsilon}^*}}\epsilon^{-1}}\\
&= 2\sqrt{N}\atomn{\m{Y}}\epsilon^{-\frac{1}{2}} + o\sbra{\tr\sbra{T\sbra{\m{u}_{\epsilon}^*}}\epsilon^{-1}}. \end{split}\label{formu:residgeq}}
Combining (\ref{formu:residleq}) and the second equality in (\ref{formu:residgeq}) yields that $\tr\sbra{T\sbra{\m{u}_{\epsilon}^*}}\epsilon^{-1}=O\sbra{\epsilon^{-\frac{1}{2}}}$. Then, the last equality in (\ref{formu:residgeq}) gives that
\equ{\cM^{\epsilon}\sbra{\m{Y}}-N\ln\epsilon\geq 2\sqrt{N}\atomn{\m{Y}}\epsilon^{-\frac{1}{2}} + o\sbra{\epsilon^{-\frac{1}{2}}}. \label{formu:residgeq2}}
The conclusion is finally drawn by combining (\ref{formu:residleq}) and (\ref{formu:residgeq2}).

\subsection{Proof of Theorem \ref{thm:epsilontozero}} \label{sec:append:epsilontozero}

Our proof is given in four steps. Let $\lbra{\lambda_{\epsilon,k}}_{k=1}^N$ be the eigenvalues of $T\sbra{\m{u}_{\epsilon}^*}$ that are sorted descendingly. In \emph{Step 1}, we attempt to show that there exists a constant $c>0$ such that $\lambda_{\epsilon,r}\geq c$ holds uniformly for $\epsilon\in(0,1]$. Let $T\sbra{\m{u}_{\epsilon}^*}=\sum_{k=1}^N\lambda_{\epsilon,k}\m{q}_{\epsilon,k}\m{q}_{\epsilon,k}^H= \m{Q}\diag\sbra{\lambda_{\epsilon,1},\dots,\lambda_{\epsilon,N}}\m{Q}^H$ be the eigen-decomposition, where $\m{q}_{\epsilon,k}$ is the $k$th column of $\m{Q}$ and $\m{Q}\m{Q}^H=\m{I}$. Then,
\equ{\cM^{\epsilon}\sbra{\m{Y}}=\sum_{k=1}^N\ln\abs{\lambda_{\epsilon,k}+\epsilon} + \sum_{k=1}^N \frac{\overline{p}_{\epsilon,k}}{\lambda_{\epsilon,k}}, \label{formu:equ1}}
where $\overline p_{\epsilon,k} \triangleq\twon{\m{q}_{\epsilon,k}^H\m{Y}}^2$. According to the optimality of $\m{u}_{\epsilon}^*$, the right hand side of the equation above obtains its minimum at $\lambda_{\epsilon,k}$. Since its derivative at $\lambda_{\epsilon,k}$ equals $\frac{1}{\lambda_{\epsilon,k}+\epsilon} -\frac{\overline{p}_{\epsilon,k}}{\lambda_{\epsilon,k}^2}$, we have that
\equ{\overline{p}_{\epsilon,k}=\left\{\begin{array}{ll} 0, & \text{ if } \lambda_{\epsilon,k}=0, \\ \frac{\lambda_{\epsilon,k}^2}{\lambda_{\epsilon,k}+\epsilon} \in \sbra{\lambda_{\epsilon,k}-\epsilon,\lambda_{\epsilon,k}}, & \text{ otherwise.} \end{array}\right. \label{formu:lambdapbar}}
Therefore,
\equ{\tr\sbra{T\sbra{\m{u}_{\epsilon}^*}} = \sum_{k=1}^N\lambda_{\epsilon,k} < \sum_{k=1}^N\overline{p}_{\epsilon,k} + N\epsilon \leq\frobn{\m{Y}}^2 + N }
provided that $\epsilon\leq 1$. It follows that $\m{u}_{\epsilon}^*$ and $\lbra{\lambda_{\epsilon,k}}$ are bounded.

On the other hand, let $T\sbra{\m{u}_{\epsilon}^*}=\sum_{k=1}^{N} p_{\epsilon,k}\m{a}\sbra{f_{\epsilon,k}}\m{a}\sbra{f_{\epsilon,k}}^H=\m{A}\m{P}\m{A}^H$ be any Vandermonde decomposition, where $\lbra{p_{\epsilon,k}}_{k=1}^N$ are sorted descendingly (note that, if $r_{\epsilon}^*=\rank\sbra{T\sbra{\m{u}_{\epsilon}^*}}<N$, then this decomposition is unique and only the first $r_{\epsilon}^*$ elements in $\lbra{p_{\epsilon,k}}_{k=1}^N$ are nonzero). Following from the fact that $\m{Y}$ lies in the range space of $\m{A}$, we have that $\m{Y}=\m{A}\m{S}$ for some $\m{S}$. Let $\m{s}_{\epsilon,k}$ be its $k$th row of $\m{S}$. Then we have
\equ{\cM^{\epsilon}\sbra{\m{Y}} = \ln\abs{\m{A}\m{P}\m{A}^H+\epsilon\m{I}} + \tr\sbra{\m{S}^H\m{P}^{-1}\m{S}}. \label{formu:equ_vanderdec}}
According to the optimality of $\m{u}_{\epsilon}^*$, the right hand side of the equation above obtains its minimum at $p_{\epsilon,k}$. As a result,
its derivative at $p_{\epsilon,k}$ equals 0, i.e.,
\equ{\m{a}\sbra{f_{\epsilon,k}}^H\sbra{\m{A}\m{P}\m{A}^H+\epsilon\m{I}}^{-1}\m{a}\sbra{f_{\epsilon,k}} - \frac{\twon{\m{s}_{\epsilon,k}}^2}{p_{\epsilon,k}^2}=0,}
and so $p_{\epsilon,k}>\twon{\m{s}_{\epsilon,k}}^2$ since
\equ{\m{a}\sbra{f_{\epsilon,k}}^H\sbra{\m{A}\m{P}\m{A}^H+\epsilon\m{I}}^{-1}\m{a}\sbra{f_{\epsilon,k}}< p_{\epsilon,k}^{-1}}
provided that $\epsilon>0$. Since $\tr\sbra{T\sbra{\m{u}_{\epsilon}^*}}=N\sum_{k=1}^{N}p_{\epsilon,k}$ and that $\m{u}_{\epsilon}^*$ is bounded as shown previously, $\lbra{p_{\epsilon,k}}$ and $\lbra{\m{s}_{\epsilon,k}}$ are bounded.

We now prove that $\lambda_{\epsilon,r}\geq c$ for some constant $c$. Otherwise, for any $c_j=\frac{1}{j}$, $j=1,2,\dots$, there exists $\epsilon_j\in(0,1]$ such that $\lambda_{\epsilon_j,r}< c_j=\frac{1}{j}$. Since the sequence $\lbra{\sbra{\m{u}_{\epsilon_j}^*, \lambda_{\epsilon_j,k}, \m{q}_{\epsilon_j,k}, p_{\epsilon_j,k}, f_{\epsilon_j,k}, \m{s}_{\epsilon_j,k}}}_{j=1}^{\infty}$ is bounded, there must exist a convergent subsequence. Without loss of generality, we assume that the sequence is convergent itself and denote by $\sbra{\m{u}^*, \lambda_{k}, \m{q}_{k}, p_{k}, f_{k}, \m{s}_{k}}$ the limit point. Since $\lambda_r=\lim_{j\rightarrow\infty}\lambda_{\epsilon_j,r}=0$, we have $\lambda_{k}=0$ for all $k=r+1,\dots,N$. It follows that
\equ{T\sbra{\m{u}^*} = \sum_{k=1}^{r-1} \lambda_k\m{q}_k\m{q}_k^H }
and $\rank\sbra{T\sbra{\m{u}^*}}\leq r-1$. As a result, at most $r-1$ $f_k$'s are retained in the decomposition $T\sbra{\m{u}^*}=\sum_{k=1}^{N} p_{k}\m{a}\sbra{f_{k}}\m{a}\sbra{f_{k}}^H$ if we remove repetitive $f_k$'s and those with $p_k=0$. Note that $\m{s}_k=\m{0}$ if $p_k=0$ since we have shown that $p_{\epsilon,k}\geq \twon{\m{s}_{\epsilon,k}}^2$. Then, by a similar operation we can reduce the order of the decomposition $\m{Y}=\sum_{k=1}^N \m{a}\sbra{f_k}\m{s}_k$ to maximally $r-1$, i.e., we obtain an atomic decomposition of order at most $r-1$, which contradicts the fact that $\norm{\m{Y}}_{\cA,0}=r$ and leads to the conclusion.

In \emph{Step 2} we prove the first part of the theorem. According to (\ref{formu:equ1}) and the bound $\lambda_{\epsilon,r}\geq c$ shown in \emph{Step 1}, we have that
\equ{\begin{split}\cM^{\epsilon}\sbra{\m{Y}}
&=\sum_{k=1}^N\ln\abs{\lambda_{\epsilon,k}+\epsilon} + \sum_{k=1}^N \frac{\overline{p}_{\epsilon,k}}{\lambda_{\epsilon,k}}\\
&\geq \sum_{k=1}^N\ln\abs{\lambda_{\epsilon,k}+\epsilon}\\
&\geq (N-r)\ln\epsilon +r\ln c. \end{split} \label{formu:Mgeqside}}
On the other hand, we consider an atomic decomposition of $\m{Y}$ of order $r$, $\m{Y}=\sum_{k=1}^r\m{a}\sbra{f_k}\m{s}_k$. Let $T\sbra{\m{u}}=\sum_{k=1}^r\twon{\m{s}_k}^2\m{a}\sbra{f_k}\m{a}\sbra{f_k}^H$, and $\lbra{\lambda_k}_{k=1}^r$ be the $r$ nonzero eigenvalues of $T\sbra{\m{u}}$. Note that $\lbra{\lambda_k}_{k=1}^r$ are constants independent of $\epsilon$. Then, provided $\epsilon\leq 1$ we have that
\equ{\begin{split}\cM^{\epsilon}\sbra{\m{Y}}
&\leq (N-r)\ln\epsilon+ \sum_{k=1}^r\ln\abs{\lambda_{k}+\epsilon} + r\\
&\leq (N-r)\ln\epsilon + \sum_{k=1}^r\ln\abs{\lambda_{k}+1} + r.\end{split} \label{formu:Mleqside}}
Combining (\ref{formu:Mgeqside}) and (\ref{formu:Mleqside}), it yields that $\cM^{\epsilon}\sbra{\m{Y}}\sim (N-r)\ln\epsilon$ as $\epsilon\rightarrow0$.

In \emph{Step 3} we prove the second part of the theorem. Based on (\ref{formu:Mgeqside}) and (\ref{formu:Mleqside}) we have that
\equ{\begin{split}
&(N-r)\ln\epsilon + c_1\\
&\geq \cM^{\epsilon}\sbra{\m{Y}} \\
&\geq \sum_{k=r+1}^N\ln\abs{\lambda_{\epsilon,k}+\epsilon} + \sum_{k=1}^r\ln\abs{\lambda_{\epsilon,k}+\epsilon}\\
&\geq (N-r)\ln\epsilon + \sum_{k=r+1}^N\ln\abs{\frac{\lambda_{\epsilon,k}}{\epsilon}+1} + c_2, \end{split}}
where $c_1$ and $c_2$ are constants independent of $\epsilon$. Therefore, it must hold that
\equ{\ln\abs{\frac{\lambda_{\epsilon,r+1}}{\epsilon}+1}\leq \sum_{k=r+1}^N\ln\abs{\frac{\lambda_{\epsilon,k}}{\epsilon}+1}\leq c_1-c_2.}
It follows that
\equ{0\leq\lambda_{\epsilon,N}\leq\dots\leq \lambda_{\epsilon,r+1}\leq \sbra{e^{c_1-c_2}-1}\epsilon,}
i.e., $\lambda_{\epsilon,k}=O\sbra{\epsilon}$, $k=r+1,\dots,N$.

Finally, we show the last part of the theorem. For any cluster point $\m{u}_0^*$ of $\m{u}_{\epsilon}^*$ at $\epsilon=0$, there exists a sequence $\lbra{\m{u}_{\epsilon_j}^*}_{j=1}^{\infty}$ converging to $\m{u}_0^*$, where $\epsilon_j\rightarrow0$ as $j\rightarrow\infty$. It must hold that $\rank\sbra{T\sbra{\m{u}_0^*}}=r$ since the smallest $N-r$ eigenvalues of $T\sbra{\m{u}_{\epsilon_j}^*}$ approach 0. Moreover, the eigen-decomposition of $T\sbra{\m{u}_{\epsilon_j}^*}$ converge to that of $T\sbra{\m{u}_0^*}$, where again we denote their eigenvalues by $\lbra{\lambda_{\epsilon_j,k}}$ and $\lbra{\lambda_{k}}$ respectively and use the other notations similarly. Then, according to (\ref{formu:lambdapbar}) we have that $\overline{p}_{\epsilon_j,k}=\frac{\lambda_{\epsilon_j,k}^2}{\lambda_{\epsilon_j,k}+\epsilon}\rightarrow \lambda_k=\overline{p}_{k}$, as $j\rightarrow\infty$. Therefore,
\equ{\tr\sbra{\m{Y}^HT\sbra{\m{u}_{0}^*}^{-1}\m{Y}} = \sum_{k=1}^r \frac{\overline{p}_k}{\lambda_k} =r. }
On the other hand, we similarly write the Vandermonde decomposition of $T\sbra{\m{u}_{\epsilon}^*}$ and let $T\sbra{\m{u}_0^*}=\sum_{k=1}^r p_k\m{a}\sbra{f_k}\m{a}\sbra{f_k}^H$, with $\m{Y}=\sum_{k=1}^r\m{a}\sbra{f_k}\m{s}_k$. It is easy to show that $p_{k}\geq\twon{\m{s}_{k}}^2$ based on the inequality $p_{\epsilon,k}>\twon{\m{s}_{\epsilon,k}}^2$, though $p_{\epsilon,k}$ and $\m{s}_{\epsilon,k}$ do not necessarily converge to $p_{k}$ and $\m{s}_{k}$ (consider the case where an accumulation point of $\lbra{f_{\epsilon_j,1},\dots,f_{\epsilon_j,N}}_{j=1}^{\infty}$ contains identical elements). Then,
\equ{\tr\sbra{\m{Y}^HT\sbra{\m{u}_{0}^*}^{-1}\m{Y}} = \sum_{k=1}^{r} \frac{\twon{\m{s}_{k}}^2}{p_{k}}\leq r, \label{formu:leqK}}
where the equality holds iff $p_{k}=\twon{\m{s}_{k}}^2$. So, we complete the proof.

\subsection{Proof of Theorem \ref{thm:weightAN}} \label{sec:Append_weightAN}
The conclusion is a direct result of the following equalities:
\equ{\begin{split}
&\min_{\m{u}} \frac{\sqrt{N}}{2}\tr\sbra{\m{W}T\sbra{\m{u}}} + \frac{1}{2\sqrt{N}}\tr\sbra{\m{Y}^HT\sbra{\m{u}}^{-1}\m{Y}},\\
&\st T\sbra{\m{u}}\geq\m{0}\\
=& \min_{f_k,p_k\geq0} \frac{\sqrt{N}}{2}\tr\sbra{\m{W}\m{R}} + \frac{1}{2\sqrt{N}}\tr\sbra{\m{Y}^H\m{R}^{-1}\m{Y}},\\
&\st \m{R}=\sum_k p_k\m{a}\sbra{f_k}\m{a}\sbra{f_k}^H\\
=& \min_{f_k,p_k\geq0,\m{s}_k} \frac{\sqrt{N}}{2}\sum_k\m{a}\sbra{f_k}^H\m{W}\m{a}\sbra{f_k}p_k\\
&\phantom{\min_{f_k,p_k\geq0,\m{s}_k}} + \frac{1}{2\sqrt{N}}\sum_k \twon{\m{s}_k}^2p_k^{-1},\\
&\st \m{Y} = \sum_k\m{a}\sbra{f_k}\m{s}_k\\
=& \min_{f_k,p_k\geq0,\m{s}_k} \frac{\sqrt{N}}{2}\sum_kw\sbra{f_k}^{-2}p_k + \frac{1}{2\sqrt{N}}\sum_k \twon{\m{s}_k}^2p_k^{-1},\\
&\st \m{Y} = \sum_k\m{a}\sbra{f_k}\m{s}_k\\
=& \min_{f_k,\m{s}_k} \sum_kw\sbra{f_k}^{-1}\twon{\m{s}_k},\st \m{Y} = \sum_k\m{a}\sbra{f_k}\m{s}_k \\
=& \norm{\m{Y}}_{\cA^w},\end{split} \label{formu:weightAN_derivation}}
where the first equality applies the Vandermonde decomposition, and the second follows the equality (see \cite{yang2014exact})
\equ{\begin{split}
&\tr\sbra{\m{Y}^H\m{R}^{-1}\m{Y}}\\
&=\min_{f_k,\m{s}_k}\sum_k \twon{\m{s}_k}^2p_k^{-1},\st \m{Y} = \sum_k\m{a}\sbra{f_k}\m{s}_k\end{split}}
given the expression of $\m{R}$ in (\ref{formu:weightAN_derivation}). Note that in the first equality in \eqref{formu:weightAN_derivation} we did not specify the order of the Vandermonde decomposition of $T(\m{u})$. This means that the proof holds true for any possible decomposition (whenever $T(\m{u})$ is invertible or not).

\subsection{Lagrangian Analysis of the Dual Problem (\ref{formu:duality})} \label{sec:append_duality}
Let $\m{\Lambda}=\begin{bmatrix}\m{U}&\m{V}^H \\ \m{V} & \m{Z}\end{bmatrix}\geq\m{0}$. The Lagrangian function of (\ref{formu:weightedSDP}) is given as follows:
\equ{\begin{split}
&\cL\sbra{\m{u},\m{X},\m{Y},\m{\Lambda},\lambda} \\
&= \tr\sbra{\m{W}T\sbra{\m{u}}} +\tr\sbra{\m{X}} - \tr\sbra{\begin{bmatrix}\m{X}&\m{Y}^H \\ \m{Y} & T\sbra{\m{u}}\end{bmatrix} \m{\Lambda}} \\ &\quad+\lambda\sbra{\frobn{\m{Y}_{\m{\Omega}} - \m{Y}_{\m{\Omega}}^o}^2 -\eta^2} \\
&= \tr\mbra{\sbra{\m{W}-\m{Z}}T\sbra{\m{u}}} + \tr\mbra{\sbra{\m{I}-\m{U}}\m{X}} -2\Re\tr\sbra{\m{Y}_{\overline{\m{\Omega}}}^H\m{V}_{\overline{\m{\Omega}}}}\\
&\quad+ \lambda\frobn{\m{Y}_{\m{\Omega}} - \m{Y}_{\m{\Omega}}^o - \lambda^{-1}\m{V}_{\m{\Omega}}}^2 -\lambda^{-1}\frobn{\m{V}_{\m{\Omega}}}^2 -\lambda\eta^2\\
&\quad - 2\Re\tr\sbra{\m{Y}_{\m{\Omega}}^{oH}\m{V}_{\m{\Omega}}}.
\end{split} \notag}
Minimizing $\cL$ with respect to $\sbra{\m{u},\m{X},\m{Y}}$ gives the dual objective which equals $-\lambda^{-1}\frobn{\m{V}_{\m{\Omega}}}^2 -\lambda\eta^2 -2\Re\tr\sbra{\m{Y}_{\m{\Omega}}^{oH} \m{V}_{\m{\Omega}}}$, if
\equ{T^*\sbra{\m{W}-\m{Z}}=\m{0},\quad \m{U}=\m{I}, \text{ and } \m{V}_{\overline{\m{\Omega}}}=\m{0}, \notag}
or $-\infty$, otherwise, where $T^*\sbra{\cdot}$ denotes the adjoint operator of $T\sbra{\cdot}$. Therefore, we obtain the dual problem in (\ref{formu:duality}) by noting that \equ{\lambda^{-1}\frobn{\m{V}_{\m{\Omega}}}^2 + \lambda\eta^2 \geq 2\eta\frobn{\m{V}_{\m{\Omega}}}. \notag}

\subsection{Proof of Proposition \ref{prop:dimreduce}} \label{sec:Append_dimreduce}
Regarding (\ref{formu:problem_j}) and (\ref{formu:problem}) we consider the following optimization problem:
\equ{\begin{split}
\min_{\m{Y}} \tr\sbra{\m{Y}^H\m{C}\m{Y}}, \st \frobn{\m{Y}_{\m{\Omega}}-\m{Y}_{\m{\Omega}}^o}^2\leq \eta^2, \end{split} \label{formu:probY}}
where $\m{C}\geq\m{0}$ is fixed. We replace the optimization variable $\m{Y}$ by $\m{Z}=\m{Y}\m{Q}$. Since $\m{Q}$ is a unitary matrix, the problem becomes
\equ{\begin{split}
\min_{\m{Z}} \tr\sbra{\m{Z}^H\m{C}\m{Z}}, \st \frobn{\m{Z}_{\m{\Omega}}-\m{Y}_{\m{\Omega}}^o\m{Q}}^2\leq \eta^2, \end{split} \notag}
and equivalently,
\equ{\begin{split}
&\min_{\m{Z}_1,\m{Z}_2} \tr\sbra{\m{Z}_1^H\m{C}\m{Z}_1}+ \tr\sbra{\m{Z}_2^H\m{C}\m{Z}_2},\\
&\st \frobn{\m{Z}_{1\m{\Omega}}-\m{Y}_{\m{\Omega}}^o\m{Q}_1}^2 + \frobn{\m{Z}_{2\m{\Omega}}}^2\leq \eta^2, \end{split} \notag}
where $\m{Z}_j=\m{Y}\m{Q}_j$, $j=1,2$. Denote the optimizer by $\sbra{\m{Z}_1^*,\m{Z}_2^*}$. It is obvious that $\m{Z}_{2}^*=\m{0}$ since $\tr\sbra{\m{Z}_2^H\m{C}\m{Z}_2}\geq0$. Then the problem becomes
\equ{\begin{split}
\min_{\m{Z}_1} \tr\sbra{\m{Z}_1^H\m{C}\m{Z}_1},\st \frobn{\m{Z}_{1\m{\Omega}}-\m{Y}_{\m{\Omega}}^o\m{Q}_1}^2\leq \eta^2, \end{split} \notag}
which is a dimension reduced version of (\ref{formu:probY}) and has the same optimal function value. Moreover, given $\m{Z}^*=\begin{bmatrix}\m{Z}_1^*&\m{0}\end{bmatrix}$ we have the optimizer to (\ref{formu:probY}) $\m{Y}^*=\m{Z}^*\m{Q}^H= \m{Z}_1^* \m{Q}_1^H$. Now the conclusion can be easily drawn.

\bibliographystyle{IEEEtran}


\end{document}